\documentclass[11pt]{article}
\usepackage{fullpage}
\usepackage{graphicx}
\usepackage{amssymb}
\usepackage{amsmath}
\usepackage{amsthm}
\newenvironment{proofofthm}[1] % #1 is the parameter for the theorem number  
  {\vspace{0.2cm}
  \par
  \noindent 
  {\em Proof of Theorem #1.}\mbox{}} % Inserts the theorem number in the title  
  {\hfill$\square$} % Closing symbol of the environment 
  {\vspace{0.2cm}
  \par
  \noindent 
  {\em Proof of Lemma #1.}\mbox{}} % Inserts the lemma number in the title  
  {\hfill$\square$} % Closing symbol of the environment 
\usepackage{tikz}

\usetikzlibrary{calc}

\usepackage[capitalise]{cleveref}

\usepackage{mathrsfs}
\newtheorem{theorem}{Theorem}[section]
\newtheorem{lemma}{Lemma}[section]

\newtheorem{observation}[theorem]{Observation}
\newtheorem{conjecture}{Conjecture}[section]
\newtheorem{corollary}[theorem]{Corollary}
\newtheorem{property}[theorem]{Property}

\theoremstyle{definition}

\newtheorem{question}{Question}

\crefname{conjecture}{Conjecture}{Conjectures}
\crefname{question}{Question}{Questions}

\newcommand{\ignore}[1]{}
             
\newcommand{\E}{{\mathbb E\/}}

\newcommand{\polylog}{\operatorname{polylog}}

\newcommand{\Err}{\operatorname{Err}}
\newcommand{\boldu}{\textbf{u}}

\newcommand{\Grid}{\mathsf{Grid}}
\newcommand{\dist}{\operatorname{dist}}
\newcommand{\monodist}{\operatorname{monodist}}
\newcommand{\mono}{\operatorname{mono}}
\newcommand{\Z}{\mathbb{Z}}

\newcommand{\R}{\mathbb{R}}
\newcommand{\Ball}{\operatorname{Ball}}
\newcommand{\Angles}{\mathrm{Angles}}
\newcommand{\Highway}{\operatorname{Highway}}

\newcommand{\Lines}{\mathsf{Lines}}
\newcommand{\Segments}{\mathsf{Segments}}
\newcommand{\Fat}{\operatorname{Fat}}

\title{The Squishy Grid Problem\thanks{Supported by NSF Grant CCF-2221980.  This research was conducted while the first three authors were visiting University of Michigan.}}

\author{Zixi Cai\\
IIIS, Tsinghua University
\and
Kuowen Chen\\
IIIS, Tsinghua University
\and 
Shengquan Du\\
IIIS, Tsinghua University
\and
Arnold Filtser\\
Bar-Ilan University
\and
Seth Pettie\\
University of Michigan
\and
Daniel Skora\\
University of Michigan}

\date{}

\begin{document}
\maketitle

\begin{abstract}
In this paper we consider the problem of approximating Euclidean distances 
by the infinite integer grid graph.  Although the topology of the graph is fixed,
we have control over the edge-weight assignment $w : E\to \R_{\geq 0}$,
and hope to have grid distances be asymptotically isometric to Euclidean distances, that is:
\[
\text{For all grid points $u,v$, $\dist_w(u,v) = (1\pm o(1))\|u-v\|_2$.}
\]
We give three methods for solving this problem, 
each attractive in its own way.
\begin{itemize}
    \item Our first construction is based on an embedding of the recursive, 
    non-periodic \emph{pinwheel tiling} of Radin and Conway~\cite{Radin94,RadinS96,ConwayR98} into the integer grid.  Distances in the pinwheel graph are asymptotically isometric to Euclidean distances, but no explicit bound on the rate of convergence was known.  We prove that the multiplicative distortion of the pinwheel graph is $(1 + 1/\Theta(\log^\xi \log D))$, where $D$ is the Euclidean distance and $\xi=\Theta(1)$.  The pinwheel tiling approach is conceptually simple, but can be improved quantitatively.

    \item Our second construction is based on a hierarchical arrangement of \emph{highways}.  It is simple, achieving stretch
    $(1 + 1/\Theta(D^{1/9}))$, which converges doubly exponentially faster 
    than the pinwheel tiling approach.

    \item The first two methods are deterministic, with rigorous guarantees.  An even simpler approach is to sample the edge weights independently and randomly from a common distribution $\mathscr{D}$.  Whether there exists a distribution $\mathscr{D}^*$ that makes grid distances Euclidean, asymptotically and in expectation, is major open problem in the theory of \emph{first passage percolation}.  Previous experiments show that when $\mathscr{D}$ is a Fisher distribution (which is continuous), grid distances are within $1\%$ of Euclidean distances.  We demonstrate experimentally that this level of accuracy can be achieved by a simple 2-point distribution that assigns weights $0.41$ or $4.75$ with probability $44$\% and $56$\%, respectively.
\end{itemize}
\end{abstract}

\section{Introduction}

In this paper we consider a natural geometric problem tangentially related to metric 
embeddings, spanners, and, in its randomized form, percolation theory.  Suppose we wish to approximate Euclidean distances between points on the plane, but with a simple \emph{discrete} structure: the integer grid graph $\Grid = (\Z\times \Z, \{\{u,v\} \mid \|u-v\|_1=1\})$.
If we consider all edges of $E(\Grid)$ to have unit length, then $\Grid$ can be regarded as a $\sqrt{2}$-spanner since for any $(u,v) \in (\Z^2)^2$, 
\[
\|u-v\|_2 \leq \dist_{\Grid}(u,v) \leq \sqrt{2}\cdot \|u-v\|_2.
\]
Now define $\Grid[w]$ to be $\Grid$ endowed with a non-negative edge-weight assignment $w : E(\Grid)\to\R_{\geq 0}$, and let $\dist_w$ be the distance function with respect to $w$.  
We consider the natural question: does there exist a $\Grid[w^*]$ that is an asymptotic 1-spanner of the Euclidean plane?

\begin{question}[The Squishy Grid Problem]\label{q:squishy-grid}
Does there exist a weight function $w^*$ such that for all $u,v\in V(\Grid)$
$\dist_{w^*}$ is 
asymptotically Euclidean?  
That is,
\[
\dist_{w^*}(u,v) = (1\pm o(1))\|u-v\|_2.
\]
If so, we may distinguish various types of convergence:
\begin{description}
    \item[\emph{Polynomial}.] $\dist_{w^*}(u,v) = \|u-v\|_2 \pm O(\|u-v\|_2)^{1-\Omega(1)}$.
    \item[\emph{Subpolynomial}.] $\dist_{w^*}(u,v) = \|u-v\|_2 \pm (\|u-v\|_2)^{o(1)}$.
    \item[\emph{Constant}.] $\dist_{w^*}(u,v) = \|u-v\|_2 \pm O(1)$.
\end{description}
\end{question}

Before discussing our approach to answering \cref{q:squishy-grid} we review the history of \cref{q:squishy-grid} and its connections to percolation theory.

\subsection{History of the Problem and Related Results}

G.~Tardos (personal communication) made us aware of a 1990 
book chapter of Pach, Pollack, and Spencer~\cite{PachPS90} 
who attributed some version of \cref{q:squishy-grid} to Paul Erd\H{o}s.  
Pach et al.~\cite{PachPS90} proved that for any fixed $\epsilon>0$
there exists a graph $G[\epsilon]$ on the vertex set $\Z\times \Z$
such that for all $u,v$,
\[
\|u-v\|_2 \leq \dist_{G[\epsilon]}(u,v) \leq (1+\epsilon)\|u-v\|_2 + O(5^{1/\epsilon}).
\]
Unfortunately, $G[\epsilon]$ is not planar and obviously depends on $\epsilon$, so it does not lead to a resolution of  \cref{q:squishy-grid}.  
Burago and Ivanov~\cite{BuragoI15} exhibited a regular, weighted graph $G$ on the vertex set $\Z\times \Z$ for which
\[
\|u-v\|_2 \leq \dist_G(u,v) \leq \|u-v\|_2 + C,
\]
for some absolute constant $C$.  However, $G$ is also not planar.

\medskip 

Borradaile and Eppstein~\cite{BorradaileE15} 
considered a more general problem: 
given a point set $P\in \R^2$, compute a weighted 
planar graph $G=(P\cup S,E)$ with \emph{Steiner points} 
$S$ such that $\dist_G(u,v)$ $(1+\epsilon)$-approximates the Euclidean distance $\|u-v\|_2$.  They proved that $|S|=O_{\epsilon,\alpha}(|P|)$ suffices, 
where $\alpha$ is the sharpest angle in the Delaunay triangulation of $P$.  A result of Chang, Krauthgamer, and Tan~\cite{ChangKT22} implies an upper bound of $O_\epsilon(|P|\polylog|P|)$, 
which is slightly superlinear but independent of $\alpha$.

\medskip 

The problem was first posed to us by G.~Bodwin, not as a deterministic design problem (\cref{q:squishy-grid}) but as a \emph{randomized} one. Whenever $\mathscr{D}$ is a distribution over $\R_{\geq 0}$, let $\Grid[\mathscr{D}]$ be the distribution of weighted graphs such that for each $e\in E(\Grid)$, $w(e)\sim \mathscr{D}$ is sampled independently from the distribution.
Is it possible to find a distribution $\mathscr{D}^*$ such that distances in $\Grid[\mathscr{D}^*]$ are Euclidean in expectation?  In more detail:

\begin{question}[Randomized Squishy Grid Problem]\label{q:randomized-squishy-grid}
Does there exist a distribution $\mathscr{D}^*$ over $\R_{\geq 0}$
such that if $\Grid[w] \sim \Grid[\mathscr{D}^*]$ is a randomly weighted graph, for all $u,v\in V(\Grid)$,
\[
\E(\dist_w(u,v)) = (1\pm o(1))\|u-v\|_2.
\]
\end{question}

The randomized process implicit in \cref{q:randomized-squishy-grid} is
actually not new, but dates back to at least a 1965 paper of 
Hammersley and Welsh~\cite{HammersleyW65}, who called it \emph{first passage percolation}. 
They imagined an \emph{orchard} in which trees were planted on the 
integer lattice.  One tree is initially infected, and the time taken for an infected tree to infect a cardinal neighbor is governed by a distribution $\mathscr{D}$ on $\R_{\geq 0}$.  
One can then ask: how far does the infection spread by time $t$? and 
what does the set of infected trees look like?

Many basic questions in first passage percolation theory remain open, 
and we can quickly summarize the known facts related to \cref{q:randomized-squishy-grid}.
Let $\mathbf{0} = (0,0)$ be the origin
and $e_\theta$ be the unit vector with angle $\theta$ degrees.
We interpret $ne_\theta$ to mean the integer point in $V(\Grid)$ nearest to $ne_\theta$.
The \emph{time constant} $\mu_0(\mathscr{D})$ 
is
such that $\lim_{n\to \infty} \dist_w(\mathbf{0}, ne_0)/n = \mu_0$ almost surely,
which exists if, 
whenever $w_1,\ldots,w_4\sim \mathscr{D}$ are independently sampled, 
$\E(\min\{w_1,w_2,w_3,w_4\}) < \infty$~\cite{Kesten86}.
It follows that $0\leq \mu_0 \leq \E(w_1\sim \mathscr{D})$, with the latter inequality
holding with equality only if $w_1\sim\mathscr{D}$ is constant almost surely~\cite{HammersleyW65}.
Similarly, the time constants for other angles 
$\mu_\theta(\mathscr{D}) = \lim_{n\to \infty} \dist_w(\mathbf{0},ne_\theta)/n$ exist, 
and collectively define the limiting \emph{shape} of the 
balls under distribution $\mathscr{D}$.
Let $B(t) = \{u\in \Z^2 \mid \dist_w(\mathbf{0},u)\leq t\}$ be the 
ball of radius $t$ around the origin.  The Cox-Durrett shape theorem~\cite{CoxD81} shows that with probability 1, 
as $t\to \infty$, $B(t)/t$ tends to a fixed \emph{limit shape} $\mathcal{B}(\mathscr{D}) \subset \R^2$.
When $\mu_0(\mathscr{D})>0$,
$\mathcal{B}(\mathscr{D})$ 
is bounded, convex, and has the same symmetries as $\Z^2$,
and when $\mu_0(\mathscr{D})=0$, $\mathcal{B}(\mathscr{D})$ is $\R^2$ itself.  See~\cite{AuffingerDH17} for an extensive survey of first passage percolation theory.

In the context of answering \cref{q:randomized-squishy-grid} we can rescale
any non-trivial distribution $\mathscr{D}$ so that its time constant $\mu_0(\mathscr{D})=1$,
i.e., distances from the origin to points on the $x$- and $y$-axes are asymptotically 
isometric.  In light of the Cox-Durrett theorem, \cref{q:randomized-squishy-grid}
asks whether there exists a $\mathscr{D}$ for which $\mathcal{B}(\mathscr{D})$
is the unit $L_2$ ball $\{x \mid \|x\|_2\leq 1\}$.

Unfortunately, there are no results characterizing $\mathcal{B}(\mathscr{D})$
for \emph{any} non-trivial distribution $\mathscr{D}$.  It is not even known
whether there exists $\mathscr{D}$ such that 
\begin{equation}\label{eqn:0-45}
\lim_{n\to \infty} 
\frac{\E(\dist_w(\mathbf{0}, ne_0))}{n} = \lim_{n\to \infty} \frac{\E(\dist_w(\mathbf{0}, ne_{45}))}{n}=1,
\end{equation}
i.e., $\mathcal{B}(\mathscr{D})$ coincides with the unit $L_2$-ball 
on the eight (inter)cardinal directions.  On the other hand, we have solid
experimental evidence that $\mathcal{B}(\mathscr{D})$ can get within 1\% of the 
unit $L_2$-ball, for certain distributions $\mathscr{D}$.
An experimental study of Alm and Deijfen~\cite{AlmD15} 
looked at various continuous distributions $\mathscr{D}$.  When $\mathscr{D}$ is the uniform distribution, the limit shape 
$\mathcal{B}(\mathscr{D})$ approximates the 
$L_2$-ball to with 4\%, whereas when $\mathscr{D}$ is exponential the limit shape is about 1.5\% away from the $L_2$-ball.  The best empirical approximation to the $L_2$-ball came from a 
Fisher distribution, with error less than 1\%.

\subsection{Results and Findings}

We provide two approaches to answering \cref{q:squishy-grid}.  
and present additional experimental evidence that \cref{q:randomized-squishy-grid} can be answered in the affirmative, using simple discrete distributions.

\medskip 

Our first construction is based on Radin and Conway's \emph{pinwheel tiling}~\cite{Radin94,RadinS96,ConwayR98}, 
a conceptually simple tiling that emerges from the observation that a right triangle with proportions $1 : 2 : \sqrt{5}$ can be partitioned into five right triangles with the same proportions.
It is known~\cite{RadinS96} that when regarded as a plane graph $G_{\operatorname{PW}}$
with edges weighted according to Euclidean distance, 
distances in the pinwheel tiling are Euclidean \emph{in the limit},
that is,
\[
\lim_{d\to \infty} \; \max_{u,v : \|u-v\|>d}\; \frac{\dist_{G_{\operatorname{PW}}}(u,v)}{\|u-v\|_2} = 1.
\]
However the rate of convergence is unknown.  We embed the pinwheel tiling into the grid graph, and prove a bound on its convergence, namely that for a constant $\xi=\Theta(1)$,
\[
\dist_{G_{\operatorname{PW}}}(u,v) = \left(1 + O\left(\frac{1}{\log^{\xi}\log \|u-v\|_2}\right)\right)\|u-v\|_2.
\]

A natural problem is to optimize the \emph{convergence} rate of the construction.  We give a new, simple construction of a weight function $w$ of the grid that is asymptotically Euclidean, with a \textbf{\emph{polynomial}} convergence rate.
\[
\dist_w(u,v) = \|u-v\|_2 + O(\|u-v\|_2^{8/9}).
\]

The construction is based on laying out ``highways'' in the plane, which are paths cleaving closely to a line with a certain slope $a$, whose edge weights are equal and chosen to approximate Euclidean distances along the highway. For example, when $a\in [0,1]$, the weights are $\frac{\sqrt{a^2+1}}{a+1}$.
In order to get a $(1+o(1))$-distance approximation, it is necessary that the set of slopes of all highways be dense in $[0,\pi)$.  Thus, there are infinitely many slopes, and infinitely many parallel highways of each slope, whose intersection pattern is quite complicated. 
The tricky part in the design stage is to decide what to do with intersecting highways.
We give a simple method that eliminates intersections while guaranteeing polynomial convergence.

\medskip 

Alm and Deijfen's~\cite{AlmD15} experimental study of first passage percolation selected $\mathscr{D}$ from 
various continuous distributions such as uniform, exponential, Gamma, and Fisher distributions.
For several of these distributions the limit shape $\mathcal{B}(\mathscr{D})$ approximated the $L_2$-ball within a few percent, with a Fisher distribution being the best.  Our experiments show that very simple distributions with support size 2 or 3 can replicate the accuracy of the continuous distributions~\cite{AlmD15}.  
For example, the improbable distribution $\mathscr{D}_2$:
\[
\Pr_{w_0 \sim \mathscr{D}_2}\left(w_0  = \left\{\begin{array}{l}
0.41401\\
4.75309
\end{array}\right.\right)
=
\left\{\begin{array}{l}
0.44273\\
0.55727
\end{array}\right.
\]
empirically approximates the Euclidean 
metric to within about $0.75\%$,
and a certain 3-point distribution $\mathscr{D}_3$ approximates it to within $0.622\%$.
A very natural question is whether other 
$L_p$ metrics can be approximated, in 
expectation, by various distributions.  
We illustrate that the uniform distribution 
and some Gamma distributions approximate 
$L_p$ metrics with $p<2$.  
None of our experiments support the 
possibility that $L_p$ metrics with $p>2$ 
can be approximated.

\subsection{Organization}

We present the construction based on pinwheel tilings in \cref{sect:pinwheel},
as well as new bounds on the convergence of the stretch of the pinwheel graph.
The highway construction is presented in \cref{sect:highway-construction}, 
having polynomial convergence.
We present the experimental findings in \cref{sect:experimental-findings}
and conclude with a discussion of several open problems related to 
\cref{q:squishy-grid,q:randomized-squishy-grid} in \cref{sect:conclusion}.

\medskip

Sections \ref{sect:pinwheel}, \ref{sect:highway-construction}, and \ref{sect:experimental-findings} are written to be entirely independent.  
One may read them in any order.

\section{A Deterministic Construction Based on Pinwheel Tilings}\label{sect:pinwheel}

The ``pinwheel'' tiling of Radin~\cite{Radin94} 
is an example of a non-periodic tiling using a single 
tile type (and its reflection).  
Let $\mathscr{T}_0,\mathscr{T}_1,\mathscr{T}_2,\ldots$ 
be a series of tilings of ever larger triangular 
swatches of the plane, and let $\mathscr{T}_\omega$ be the tesselation 
of the plane achieved in the limit.  $\mathscr{T}_0$ 
consists of a single right triangle with side lengths $1,2,\sqrt{5}$.  In general $\mathscr{T}_{i+1}$
is formed from $\mathscr{T}_i$ by taking four
additional copies of $\mathscr{T}_{i}$, suitably
reflected, rotated, and translated, 
so that they form a larger triangle with the same 
$1:2:\sqrt{5}$ proportions.
\cref{fig:pinwheel-tiling} illustrates the construction
of $\mathscr{T}_2$ from $\mathscr{T}_1$ and $\mathscr{T}_0$.

\subsection{Pinwheel Tilings}

By construction $\mathscr{T}_\omega$ is a tiling of the 
plane using \emph{atomic} triangles with side lengths $1,2,\sqrt{5}$.  Due to the recursive nature of the construction, we can also regard $\mathscr{T}_\omega$ as a tiling using $\sqrt{5}^i,2\sqrt{5}^i,\sqrt{5}^{i+1}$ triangles, for any integer $i\geq 0$.  Observe from \cref{fig:pinwheel-tiling}
that the boundary of $\mathscr{T}_{i+1}$ is obtained from the boundary $\mathscr{T}_i$ by scaling by $\sqrt{5}$, translation, and rotation by $\arctan(1/2)$.  We will henceforth define $\gamma=\arctan(1/2)$.  As $\gamma/(2\pi)$ is irrational, the orientation of tiles in $\mathscr{T}_\omega$ is uniformly distributed in $[0,2\pi)$.  Radin and Sadun~\cite{RadinS96} used this fact to prove an isoperimetric property of 
pinwheel tilings, namely that there are finite subsets
of tiles from $\mathscr{T}_\omega$ whose area/perimeter$^2$
is arbitrarily close to that of the circle.
Suppose we regard $\mathscr{T}_\omega$ as a plane graph $G_\omega$, whose vertices and edges are the union of the vertices and edges of all atomic triangles.  Radin and Sadun~\cite{RadinS96} proved that for $u,v\in V(G_{\omega})$, 
$\dist_{G_\omega}(u,v) = (1+o(1))\|u-v\|_2$.  Although the multiplicative stretch is 1 \emph{in the limit}, their proof implies no particular rate of convergence.  We prove the following.

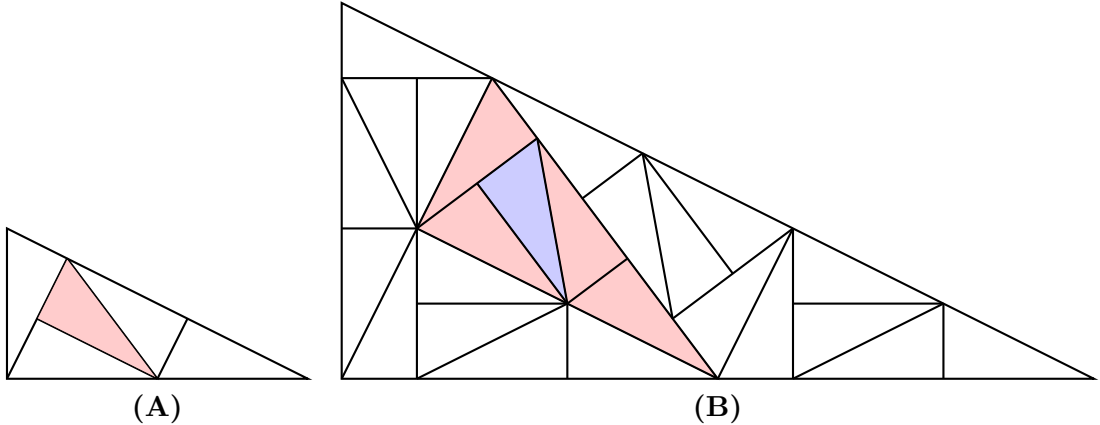
\begin{figure}
\begin{tabular}{cc}
\begin{tikzpicture}[scale=2]

% Coordinates of the big triangle
\coordinate (A) at (0, 0);
\coordinate (B) at (2, 0);
\coordinate (C) at (0, 1);

% Draw outer triangle
\draw[thick] (A) -- (B) -- (C) -- cycle;

% Compute interior points for subdivision (standard pinwheel substitution)
% The key idea is to rotate and scale the triangles correctly
% Reference: Radin's subdivision rule

% Point D divides AB in ratio 2:1
\coordinate (D) at ($(A)!1/2!(B)$);
\coordinate (E) at ($(C)!3/5!(B)$);
\coordinate (F) at ($(C)!1/5!(B)$);
\coordinate (G) at ($(A)!1/2!(F)$);

% Draw the five internal triangles
\draw[thick] (A) -- (F);
\draw[thick] (D) -- (E);
\draw[thick] (F) -- (D);
\draw[thick] (G) -- (D);
\draw[black, fill=red!20] (D) -- (F) -- (G) -- cycle;

\end{tikzpicture}

&

\begin{tikzpicture}[scale=5]

% Coordinates of the big triangle
\coordinate (A) at (0, 0);
\coordinate (B) at (2, 0);
\coordinate (C) at (0, 1);

% Draw outer triangle
\draw[thick] (A) -- (B) -- (C) -- cycle;

% Compute interior points for subdivision (standard pinwheel substitution)
% The key idea is to rotate and scale the triangles correctly
% Reference: Radin's subdivision rule

% Point D divides AB in ratio 2:1
\coordinate (D) at ($(A)!1/2!(B)$);
\coordinate (E) at ($(C)!3/5!(B)$);
\coordinate (F) at ($(C)!1/5!(B)$);
\coordinate (G) at ($(A)!1/2!(F)$);
\coordinate (H) at ($(D)!3/5!(B)$);
\coordinate (J) at ($(D)!1/5!(B)$);
\coordinate (I) at ($(E)!1/2!(B)$);
\coordinate (K) at ($(E)!1/2!(J)$);
\coordinate (L) at ($(D)!1/5!(F)$);
\coordinate (M) at ($(D)!3/5!(F)$);
\coordinate (N) at ($(E)!1/2!(F)$);
\coordinate (O) at ($(E)!1/2!(L)$);
\coordinate (P) at ($(A)!3/5!(D)$);
\coordinate (Q) at ($(A)!1/5!(D)$);
\coordinate (R) at ($(G)!1/2!(Q)$);
\coordinate (S) at ($(G)!1/2!(D)$);
\coordinate (T) at ($(F)!1/5!(D)$);
\coordinate (U) at ($(G)!1/2!(T)$);
\coordinate (V) at ($(F)!3/5!(D)$);
\coordinate (W) at ($(C)!1/5!(A)$);
\coordinate (X) at ($(C)!3/5!(A)$);
\coordinate (Y) at ($(W)!1/2!(F)$);

\draw[black, fill=red!20] (D) -- (F) -- (G) -- cycle;
\draw[black, fill=blue!20] (S) -- (U) -- (T) -- cycle;

% Draw the five internal triangles
\draw[thick] (A) -- (F);
\draw[thick] (D) -- (E);
\draw[thick] (F) -- (D);
\draw[thick] (G) -- (D);

\draw[thick] (J) -- (I);
\draw[thick] (J) -- (E);
\draw[thick] (H) -- (I);
\draw[thick] (K) -- (I);

\draw[thick] (E) -- (L);
\draw[thick] (M) -- (N);
\draw[thick] (L) -- (N);
\draw[thick] (O) -- (N);

\draw[thick] (G) -- (T);
\draw[thick] (S) -- (U);
\draw[thick] (S) -- (T);
\draw[thick] (S) -- (V);

\draw[thick] (Q) -- (G);
\draw[thick] (S) -- (R);
\draw[thick] (S) -- (Q);
\draw[thick] (S) -- (P);

\draw[thick] (F) -- (W);
\draw[thick] (G) -- (Y);
\draw[thick] (G) -- (W);
\draw[thick] (G) -- (X);

\end{tikzpicture}\\
\textbf{(A)} & \textbf{(B)}
\end{tabular}
\caption{\label{fig:pinwheel-tiling}
\textbf{(A)} $\mathscr{T}_1$, containing $\mathscr{T}_0$ in red.
\textbf{(B)} $\mathscr{T}_2$, containing $\mathscr{T}_1$ in red, and $\mathscr{T}_0$ in blue.}

\end{figure}

\begin{theorem}\label{thm:pinwheel-stretch-1}
    Let $G_\omega$ be the plane graph of the pinwheel tiling $\mathscr{T}_\omega$, whose edges are weighted according to the 
    Euclidean distance between their endpoints.  Then for any $u,v\in V(G_{\omega})$,
    \[
    \|u-v\|_2 \leq \dist_{G_\omega}(u,v) \leq  (1+O(1/(\log\log\|u-v\|_2)^\xi))\cdot \|u-v\|_2,
    \]
    for some $\xi > 0$.
\end{theorem}

\subsection{Distribution of Tile Orientations}

If $T$ is a triangle in the recursive tiling with dimensions
$\sqrt{5}^i,2\sqrt{5}^i,\sqrt{5}^{i+1}$, we call $T$ a \emph{level-$i$} triangle.

As a first step toward proving \cref{thm:pinwheel-stretch-1}, we analyze the orientations of the triangles contained within a single large triangle.  Given a triangle $T$ of level $x$, let $\Angles(T,k)$ denote the set of angles attained by the hypotenuses of all level-$(x-k)$ triangles contained within $T$.  We observe that the elements of this set are characterized by an arithmetic recurrence.

\begin{observation}\label{obs:rotations}
    Suppose a triangle $T$ of level $x$ has its hypotenuse at angle $\theta$.  For $0\leq k\leq x$, $\Angles(T,k) \supseteq \left\{\theta + (2t-k)\gamma \,|\, t \in\{0,\ldots,k\}\right\}$.
\end{observation}

\begin{proof}
    We proceed by induction on $k$, with the base case $k=0$ being trivial.   Consult \cref{fig:rotations}, where two triangles $A$ and $B$ of level-$(x-1)$ are depicted.  Observe that the hypotenuse of $A$ (resp. $B$) is rotated by an angle of $-\gamma$ (resp. $+\gamma$) relative to that of the exterior triangle.  By the inductive hypothesis,
    \begin{align*}
    \Angles(A,k-1) &\supseteq \{\theta -\gamma + (2t-k+1)\gamma \,\,|\,\, t \in0,\ldots,k-1\} = \{\theta + (2t-k)\gamma \,\,|\,\, t \in0,\ldots,k-1\}\\
    \Angles(B,k-1) &\supseteq \{\theta + \gamma + (2t-k+1)\gamma \,\,|\,\, t \in0,\ldots,k-1\} = \{\theta + (2t-k)\gamma \,\,|\,\, t \in1,\ldots,k\},
    \end{align*}
    and the union of these sets is exactly $\left\{\theta + (2t-k)\gamma \,|\, t \in\{0,\ldots,k\}\right\}$.

\end{proof}

\begin{figure}
\centering
\begin{tikzpicture}[scale=2.5]

% Coordinates
\coordinate (A) at (0, 0);
\coordinate (B) at (2, 0);
\coordinate (C) at (0, 1);

% Draw outer triangle
\draw[thick] (A) -- (B) -- (C) -- cycle;

% Point D divides AB in ratio 2:1
\coordinate (D) at ($(A)!1/2!(B)$);
\coordinate (E) at ($(C)!3/5!(B)        $);
\coordinate (F) at ($(C)!1/5!(B)$);
\coordinate (G) at ($(A)!1/2!(F)$);

% Draw the five internal triangles
\draw[thick] (A) -- (F);
\draw[thick] (D) -- (E);
\draw[thick] (F) -- (D);
\draw[thick] (G) -- (D);

\draw[black, fill=red!20] (D) -- (F) -- (G) -- cycle;
\draw[black, fill=blue!20] (G) -- (D) -- (A) -- cycle;

\node at (0.3,0.183) {B};
\node at (0.433,0.45) {A};
\end{tikzpicture}
\caption{\label{fig:rotations} Two triangles within $T$. 
The hypotenuses of $T,A,B$ have angles $-\gamma,-2\gamma,0$, respectively.}
\end{figure}

Consider some $\lambda\in \mathbb{R}$, and denote by 
$\{\lambda\} = \lambda-\lfloor \lambda \rfloor$
the fractional part of $\lambda$.  
It is well known due to Weyl's criterion that if $\lambda$ is irrational, the set $S(\lambda,N)=\left\{\{kx\} : 1\leq k \leq N \right\}$ becomes uniformly distributed in $[0,1)$ as $N\to \infty$~\cite{weyl-thm}.  
Moreover, the rate of convergence is controlled by the \emph{irrationality exponent} $\mu(\lambda)$ which measures the asymptotic quality of rational approximations to $\lambda$.  

For a real number $\lambda$, $\mu(\lambda)$ is defined to be the supremum of the set of real numbers $\mu$ such that
\[
0<\left|\lambda - \frac{p}{q}\right| < \frac{1}{q^\mu}
\]
has only finitely many solutions for positive integers $p,q$.  A finite irrationality exponent implies that a number does not have a sequence of
rational approximations that are ``too good.''  For sets $S\subseteq U\subseteq\mathbb{R}$, we say $S$ is an $\varepsilon$\emph{-cover} of $U$ if for every $x_1\in U$ there exists $x_2\in S$ satisfying $|x_1-x_2|<\varepsilon$.  When $\mu(\lambda)$ is finite, $S(\lambda,N)$ is an $\varepsilon$-cover of $[0,1)$ for $N=O(\varepsilon^{1-\mu(\lambda)-o(1)})$~\cite{epsilon-cover}.  The following corollary follows by applying this fact to \cref{obs:rotations}.

\begin{corollary}\label{cor:uniform-dist}
    Given a triangle $T$ of level-$x$,  $\Angles(T,k)$ 
    is a $\theta$-cover of $[0,2\pi)$ for $k=O(\theta^{1-\mu(\gamma/\pi)-o(1)})$. % >x-k
\end{corollary}

Before bounding the irrationality exponent $\mu(\gamma/\pi)$, we quickly review some terminology related to the algebraic numbers $\overline{\mathbb{Q}}$.  For $\alpha\in\overline{\mathbb{Q}}$, its \emph{minimal polynomial} is the unique polynomial $P\in \mathbb{Z}[x]$ of lowest degree with relatively prime coefficients such that $P(\alpha)=0$.  We say the \emph{degree} of $\alpha$ is the degree of $P$, while the \emph{height} of $\alpha$ is the absolute value over coefficients of $P$.

Both $\gamma$ and $\pi$ can be expressed in the form $\beta \ln\alpha$ where $\alpha\in \overline{\mathbb{Q}}$ and $\beta\in \mathbb{Q}(i)$.  Since $e^{i\gamma}=\frac{2+i}{\sqrt{5}}$, $\gamma = -i\ln{(\frac{2+i}{\sqrt{5}})}$.  Similarly, $\pi$ can be written as $-i\ln{(-1)}$.  A great deal of work in the mid-20th century yielded various lower bounds on linear forms in logarithms of algebraic numbers.  %Gelfond Hilbert 7th, Baker generalized then improved bound, Baker's survey
A history of the problem up to 1976 can be found in~\cite{baker1977-exposition}.  
For our purposes, we choose a simple bound due to Baker.

\begin{theorem}[Baker~\cite{baker1977-exposition}]\label{thm:baker}
    Suppose for $n\geq 1$ we have algebraic numbers $\alpha_1,\ldots, \alpha_n, \beta_1, \dots, \beta_n\in \overline{\mathbb{Q}}\setminus 0$.  If the logarithms $\ln{\alpha_i}$ are linearly independent over the rational numbers, then
    \[
    \left|\beta_1\ln{\alpha_1} + \cdots + \beta_n\ln{\alpha_n}\right|>H^{-C},
    \]
    where $H$ is the maximum of the heights of the $\beta_i$ and $C$ is a function of $n$, the numbers $\alpha_i$, and the degrees of the numbers $\beta_i$.
\end{theorem}

Consider arithmetic expressions of the form $\beta_1\ln{\alpha_1} + \beta_2\ln{\alpha_2}$ where $\alpha_1, \alpha_2\in \overline{\mathbb{Q}}\setminus 0$ are fixed and $\beta_1, \beta_2\in \mathbb{Z}\setminus 0$.  
Now $\beta_1$ and $\beta_2$ have degree $1$, so the exponent $C$ becomes a constant, and $H=\max\{|\beta_1|, |\beta_2|\}$.  Applying Baker's theorem (\cref{thm:baker}) 
and dividing both sides by $|\beta_1\ln{\alpha_2}|$ yields

\[
\left|\frac{\ln{\alpha_1}}{\ln{\alpha_2}} + \frac{\beta_2}{\beta_1}\right| > \frac{1}{|\beta_1\ln{\alpha_2}|H^C} \geq \frac{1}{|\ln{\alpha_2}|H^{C+1}}.
\]

By choosing $\alpha_1 = \frac{2+i}{\sqrt{5}}$ and $\alpha_2 = -1$, the ratio $\frac{\ln{\alpha_1}}{\ln{\alpha_2}}$ is exactly $\gamma/\pi$.  Recall that this number is irrational, satisfying the criteria that $\ln{\alpha_1}$ and $\ln{\alpha_2}$ are linearly independent over the rationals.  Since $\gamma/\pi<1$, clearly we can replace $H$ by $\beta_1$ to obtain $\left|\frac{\ln{\alpha_1}}{\ln{\alpha_2}} + \frac{\beta_2}{\beta_1}\right| > \frac{1}{|\ln{\alpha_2}|\beta_1^{C+1}}$.  Comparing this with the definition of the irrationality exponent, this is sufficient to see that $\mu(\gamma/\pi)$ is at most $C+1$ and therefore finite.

Henceforth let $\mu = \mu(\gamma/\pi)$.  It is worth noting that the constant $C$ given by \cref{thm:baker} is effectively computable, though the order of magnitude is impractical.\footnote{One estimate, also due to Baker~\cite{baker1977-exposition}, states that $C=\ln{A_1}\ln^2{A_2}(16nd)^{200n}$ suffices when the $\beta_i$ are rational, where $A_i$ is the height of $\alpha_i$ and $d$ is the degree of the field extension $\mathbb{Q}[\alpha_1, \alpha_2]/\mathbb{Q}$.  For our purposes, $A_1=4$, $A_2=2$, $n=2$, and $d=4$, giving $C\approx2^{2800}$.  Baker remarks it is possible to argue $C$ is much lower in reality, but the required analysis is quite technical and beyond the scope of this paper.}

\subsection{Convergence of Stretch}

We define $f(d)$ to be the maximum stretch guaranteed by $G_\omega$, over all pairs of vertices at Euclidean distance at least $d$.

\[
f(d) = \sup\left\{\frac{\dist_{G_\omega}(u,v)}{\|u-v\|_2} \;\middle|\; u,v\in V(G_{\omega}) \text{ and } \|u-v\|_2\geq d\right\}.
\]

\begin{lemma}\label{lem:pinwheel-recurrence}
    Fix any distance $d$ and let $f(d)=1+\epsilon$.
    If $d = \Omega(5^{\Theta(\varepsilon^{(1-\mu)/2})})$, then 
    \[
    f(3d) < 1+\epsilon - \Omega\left(5^{-\Theta(\varepsilon^{(1-\mu)/2})}\right).
    \]
\end{lemma}

\begin{proof}

    Consider any two vertices $P,Q\in G_\omega$ with $|\overline{PQ}|=D\geq3d$ when we regard them as points in the plane.  %We maintain and tune auxiliary parameters $\theta$ and $\delta$.
    Begin by choosing a parameter $\delta=\delta(\epsilon)$ and identifying a triangle $T$ satisfying the following properties: $T$ intersects $\overline{PQ}$, the projection of $T$ onto $\overline{PQ}$ lies entirely within the middle third of $\overline{PQ}$, and the hypotenuse of $T$ has length exactly $\delta D$.  Now choose a parameter $\theta=\theta(\epsilon)$ and let 
    $n=n(\theta)$ be large enough such that $\Angles(T,n)$ is a $\theta$-cover of $[0,\pi)$.  Finally, select a triangle $T'$ within $T$, and exactly $n$ levels below $T$, whose hypotenuse creates an angle less than $\theta$ with $\overline{PQ}$.

\medskip
    \centerline{\includegraphics[width=0.95\linewidth]{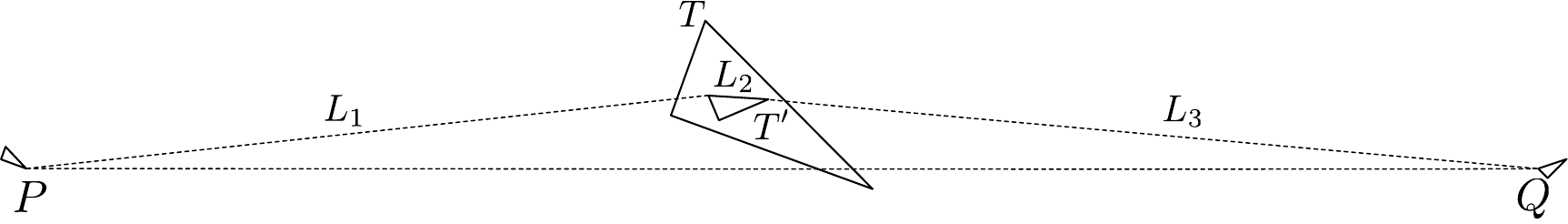}}
\medskip 

    We now construct a path as follows.  Let $p$ and $q$ be the endpoints of the hypotenuse of $T'$ closer to $P$ and $Q$ respectively.  Let $L_1=\overline{Pp}$, $L_2 = \overline{pq}$, and $L_3=\overline{Qq}$.  $|L_1|$ and $|L_3|$ must each be at least $d$, so 
    by the definition of the stretch function $f$,
    $G_\omega$ approximates $L_1$ and $L_3$ to within $1+\varepsilon$ stretch.

    Let $\ell_1$, $\ell_2$, and $\ell_3$ be the projections of $L_1$, $L_2$, and $L_3$
    onto $\overline{PQ}$.  We claim that $\dist_{G_\omega}(P,p)/|\ell_1| < (1+\varepsilon)(1+5\delta^2)$.  Since $p$ lies within $T$, its projection onto $\overline{PQ}$ is a distance at most $\delta D$ from $p$ and at least $D/3$ from $P$. Therefore, $|L_1|/|\ell_1|$ is bounded by $\sqrt{(D/3)^2+(\delta D) ^2}/(D/3) = \sqrt{1+9\delta^2}<1+5\delta^2$.  By definition of $f$, 
    $\dist_{G_\omega}(P,p)/|L_1| \leq 1+\varepsilon$.  The same is true for $\dist_{G_\omega}(q,Q)/|\ell_3|$.  On the other hand, $\dist_{G_\omega}(p,q)/|\ell_2|$ is trivially bounded by $|L_2|/|\ell_2|\leq 1/\cos\theta$ since $p$ and $q$ are joined directly by an edge in $G_\omega$.  Combining the bounds on each component, we write

    \[\dist_{G_\omega}(P,Q) <  (1+\varepsilon) (1+5\delta^2)(|\ell_1|+|\ell_3|) + \frac{|\ell_2|}{\cos\theta}.\]

   We normalize by $1/D$ to obtain the stretch, and because this inequality holds for arbitrary $P,Q$ satisfying $|\overline{PQ}|\geq 3d$, it holds for the supremum as well.
    
    \begin{align*}
        f(3d) &= \sup_{P,Q\in V(G_{\omega})}\left\{\frac{\dist_{G_\omega}(P,Q)}{|\overline{PQ}|}  \;\middle|\; |\overline{PQ}|\geq 3d\right\}\\
        &< (1+\varepsilon)(1+5\delta^2)\left(1-\frac{|\ell_2|}{D}\right) + \frac{|\ell_2|}{D\cos\theta}\\
        &\leq (1+\varepsilon)(1+5\delta^2)(1-5^{-n/2}\delta \cos\theta) + 5^{-n/2}\delta
\intertext{Letting $\kappa = 5^{-n/2}$, this is upper bounded by}
        &< 1+\varepsilon + 6\delta^2 + \kappa\delta(1 - (1+\varepsilon)\cos\theta) 
\intertext{At this point we fix $\theta=\sqrt{\epsilon}$, so $\cos\theta > 1-\epsilon/2$.} 
        &< 1+\varepsilon + 6\delta^2 - \kappa\delta\epsilon(1-\epsilon)/2
        = 1+\varepsilon + \delta(6\delta - \kappa\epsilon(1-\epsilon)/2)\\
\intertext{Finally, we pick $\delta < \kappa\epsilon/24$.}
        &\leq 1+\varepsilon - \Omega(\delta^2)
        = 1+\epsilon - \Omega(5^{-n}\epsilon^2)
        = 1+\epsilon - \Omega(5^{\Theta(\epsilon^{(1-\mu)/2})}).
    \end{align*}
The last line follows from \cref{cor:uniform-dist}, which states 
$n(\theta)=n(\sqrt{\epsilon})=(\sqrt{\epsilon})^{1-\mu}$.
    
    The only remaining detail is to address that the triangles $T$ and $T'$ used in this argument do indeed exist.   The length of the hypotenuse of $T$ was taken to be $\delta D$, and then we chose $T'$ to be $n$ levels beneath that of $T$.  Therefore, we have implicitly assumed $D\geq5^{n/2}/\delta > \Omega(5^{n}/\varepsilon)$, consistent with the assumption in the statement of the lemma.
\end{proof}

\cref{thm:pinwheel-stretch-1} now follows easily from this lemma.

\begin{proof}
    Take some initial constants $D_0,\varepsilon_0$ with $f(D_0)=1+\varepsilon_0$ and consider the sequence $\{\varepsilon_i\}$ where $\varepsilon_i = f(3^{i}D_0)-1$.  
    Define $i_{\operatorname{half}}$ to be the minimum value
    such that 
    $\varepsilon_{i_{\operatorname{half}}}\leq \varepsilon_0/2$.
    By \cref{lem:pinwheel-recurrence}, if $D$ is sufficiently large, we have that $\varepsilon_i-\varepsilon_{i+1}>\Omega\left( 5^{-\Theta\left(\varepsilon_i^{(1-\mu)/2}\right)}\right)
    >
    \Omega\left( 5^{-\Theta\left(\varepsilon_0^{(1-\mu)/2}\right)}\right)$ as long as $\varepsilon_i>\varepsilon_0/2$.
    Therefore, $i_{\operatorname{half}} = O\left(5^{\Theta\left(\varepsilon_0^{(1-\mu)/2}\right)}\right)$.
    If our target stretch is $1+\varepsilon$
    we can apply this halving argument $k=\log\epsilon^{-1}$ times, implying 
    $f(3^{i^*}D_0) < 1+\varepsilon$ for 
    \[
    i^* = O\left(5^{\Theta\left(\varepsilon_0^{(1-\mu)/2}\right)} + 
    5^{\Theta\left((\varepsilon_0/2)^{(1-\mu)/2}\right)}
    +\cdots+
    5^{\Theta\left((\varepsilon_0/2^k)^{(1-\mu)/2}\right)}\right)
    =
    O\left(5^{\Theta\left(\varepsilon^{(1-\mu)/2}\right)}\right).
    \]
    
Note that for $D=3^{i^*}D_0$, 
$\epsilon = \Theta(\log^{\xi}\log D)$, for $\xi=2/(1-\mu) = \Theta(1)$.
\end{proof}

\subsection{Pinwheel Tilings on the Grid}

\begin{theorem}\label{thm:pinwheel-grid-stretch-1}
    There exists a weight function $w$ such that 
    $\Grid[w]$ has stretch 1, asymptotically.  
    In particular, for any $u,v\in V(\Grid)$,
    \[
    \|u-v\|_2 - O(1) \leq \dist_{w}(u,v) \leq  (1+O(1/(\log\log\|u-v\|_2)^\xi))\cdot \|u-v\|_2,
    \]
    for some $\xi > 0$.
\end{theorem}

\begin{proof}
    Regard $G_{\omega}$ as the plane graph of a tiling whose 
    atomic tile has large side lengths, say $25,50,25\sqrt{5}$.
    The vertices of $G_{\omega}$ generally do not have integer coordinates.
    Let $\phi : V(G_{\omega}) \to V(\Grid)$ map 
    any $u\in V(G_{\omega})$ to the nearest integer point 
    $\phi(u)\in \Z^2$.  
    We will overload this notation a bit and let $\phi : E(G_{\omega}) \to 2^{E(\Grid)}$ be such that $\phi(\{u,v\})$ is a monotone 
    path in $\Grid$ connecting $\phi(u)$ to $\phi(v)$ cleaving closely
    to the $\overline{uv}$ line segment, with the property
    that any two paths $\phi(\{u,v\}),\phi(\{u,v'\})$ only 
    intersect in a prefix of at most 2 edges, 
    and any $\phi(\{u,v\}),\phi(\{u',v'\})$ 
    (with $u,v,u',v'$ distinct) do not intersect at all.  
    %See \cref{fig:pinwheel-grid}.  

\medskip
%\begin{figure}[h!]
\centerline{\includegraphics[scale=.2]{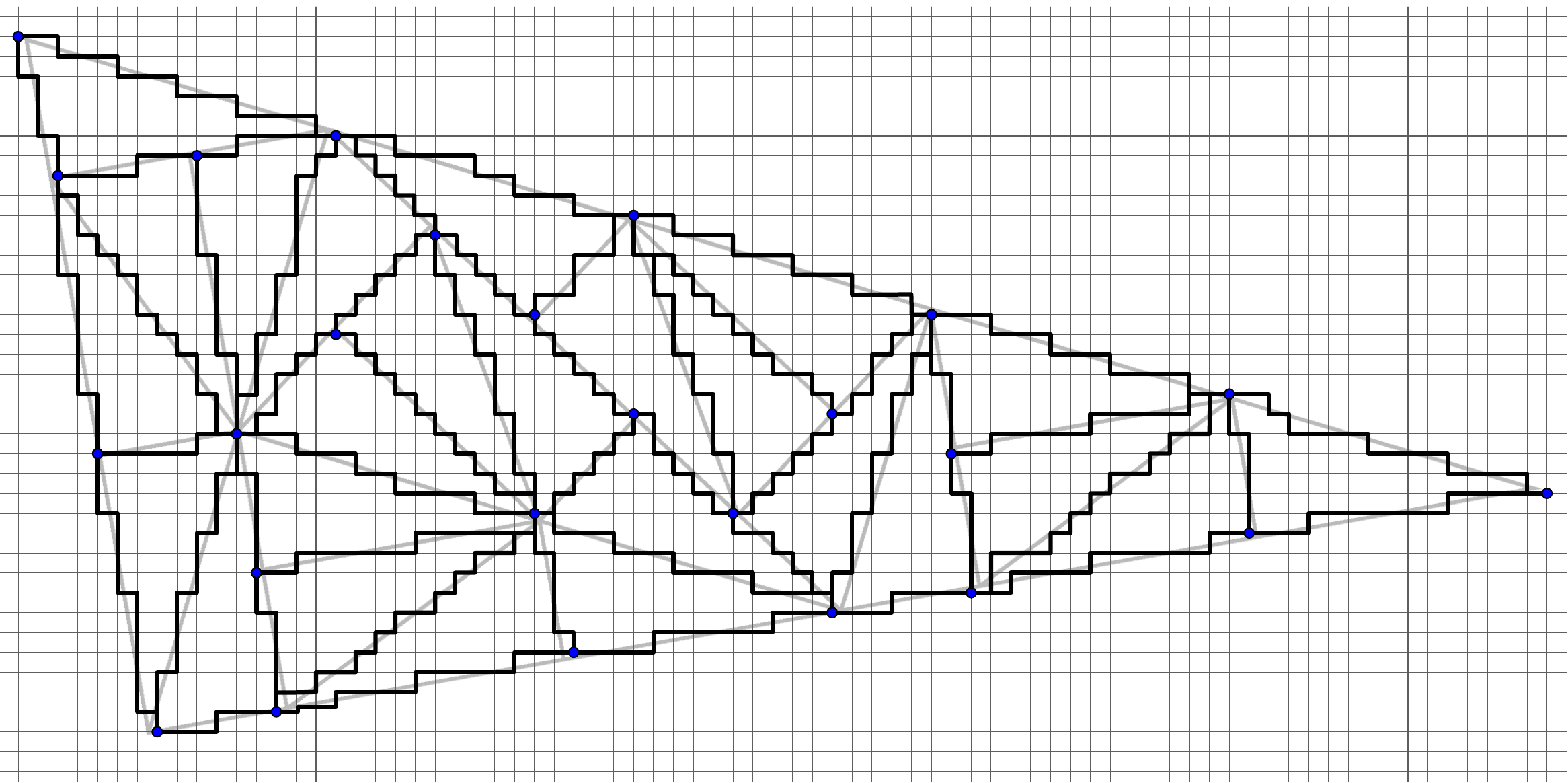}}
%\caption{\label{fig:pinwheel-grid}An embedding of the pinwheel tiling into the integer grid.}
%\end{figure}
\medskip

    The edge weights are assigned as follows.  If $e$ is not in $\bigcup_{e'\in E(G_{\omega})} \phi(e')$ then $w(e) = 10$.
    If $e$ is in two distinct paths $\phi(e'),\phi(e'')$ then $w(e)=1$. 
    The remaining edge weights of $\phi(\{u,v\})$ are chosen to be 
    equal, such that
    \[ 
        w(\phi(\{u,v\})) = \sum_{e\in \phi(\{u,v\})} w(e) = \|u-v\|_2.
    \]
    In other words, walking from $\phi(u)$ to $\phi(v)$ along $\phi(\{u,v\})$ is precisely the Euclidean distance $\|u-v\|_2$.  Depending on the angle of the $u$-$v$ line, the ``ideal'' weight of edges on $\phi(\{u,v\})$ is in the range $[1/\sqrt{2},1]$, but the true weights lie in the range $[0.6,1.05]$.  The internal edges of $\phi(\{u,v\})$ may need to have weight less than $1/\sqrt{2}$ due to rounding $u,v$ to farther integer points $\phi(u),\phi(v)$, and correcting for up to four edges on the ends of $\phi(\{u,v\})$ with weight 1.  Similarly, the internal edges of $\phi(\{u,v\})$ may need to have weight greater than 1 due to rounding $u,v$ to closer integer points $\phi(u),\phi(v)$.  However, one may verify that the length of every subpath of $\phi(\{u,v\})$ 
    from $u'$ to $v'$ differs from its Euclidean length $\|u'-v'\|_2$ by at most 2.

    By \cref{thm:pinwheel-stretch-1}, for any $u,v\in V(\Grid)$, $\dist_w(u,v) \leq (1+O(1/(\log\log\|u-v\|_2)^\xi))\cdot \|u-v\|_2$.
    One walks from $u$ to a nearby $\phi(u_0)$ vertex, then along embedded paths of the pinwheel graph $G_\omega$ to a $\phi(v_0)$ near $v$, then along a path from $\phi(v_0)$ to $v$.  By design, the length of the path from $\phi(u_0)$ to $\phi(v_0)$ is precisely $\dist_{G_\omega}(u_0,v_0)$, while the $u$-$u_0$ and $v_0$-$v$ paths have length $O(1)$.
    The weight of edges outside of $\bigcup_{e\in E(G_\omega)} \phi(e)$ is set sufficiently high so that it is never advantageous to use them in lieu of paths in $\bigcup_{e\in E(G_\omega)} \phi(e)$.
\end{proof}

\section{A Deterministic Construction with Faster Convergence}\label{sect:highway-construction}

In this section, we give a deterministic construction based on \emph{highways} with faster convergence.
Specifically, we establish the following result as \cref{thm:highway1}:

\begin{theorem}\label{thm:highway1}
There exists an assignment
$W : E(\Grid)\to \mathbb{R}_{\geq 0}$ 
such that for any $u,v\in V(\Grid)$,
\[
\|u - v\|_2 - 1 \le \dist_W(u, v) \le \|u - v\|_2 + O\left(\|u - v\|_2^{\frac{8}{9}}\right).
\]
\end{theorem}

We prove \cref{thm:highway1} in two steps.  First, we show that the same statement holds for the finite square grid $[n]\times [n]$ (\cref{thm:highway2}), 
then we give a black-box reduction from the 
infinite case to the finite case.

\begin{theorem}
\label{thm:highway2}
There exists a weight assignment $W(n)$ to edges of the finite grid on $[n]\times[n]$ such that
for any $u, v\in [n]^2$,
\[
\|u - v\|_2 - 1 \le \dist_{W(n)}(u, v) \le \|u - v\|_2 + O\left(\|u - v\|_2^{\frac{8}{9}}\right).
\]
\end{theorem}
The proof of \cref{thm:highway2} follows from the constructions presented in \cref{subsec:highways} and \cref{subsec:hierarchical-construction}.

\subsection{Highways}
\label{subsec:highways}

Given a line $\ell = ax +b$ in $\mathbb{R}^2$
we define the $\Highway(\ell)$ to be a grid-path 
that tracks $\ell$; see \cref{fig:highway}.

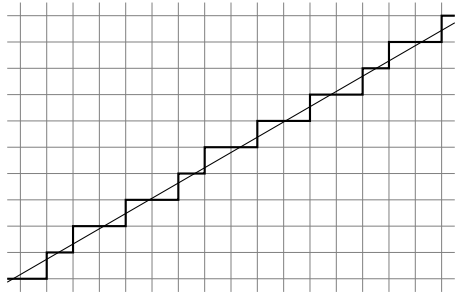
\begin{figure}
    \begin{center}
    \begin{tikzpicture}[scale=0.35]
      \def\xmin{0.5}
      \def\xmax{17.5}
      \def\ymin{-0.5}
      \def\ymax{10.5}
    
      \foreach \y in {0,1,...,10} {
        \draw[gray, very thin] (\xmin,\y) -- (\xmax,\y);
      }
      \foreach \x in {1,2,...,17} {
        \draw[gray, very thin] (\x,\ymin) -- (\x,\ymax);
      }

      \draw[black] (0.5,-0.13) -- (17.5,9.73);

      \draw[thick, black] (0.5,0) -- (2,0) -- (2,1) -- (3,1) -- (3,2) -- (5,2) -- (5,3) -- (7,3) -- (7,4) -- (8,4) -- (8,5) -- (10,5) -- (10,6) -- (12,6) -- (12, 7) -- (14,7) -- (14, 8) -- (15, 8) -- (15,9) -- (17,9) -- (17,10) -- (17.5,10);    
    \end{tikzpicture}
    \end{center}
    \caption{\label{fig:highway}A grid-path that tracks $\ell$.}
\end{figure}

Specifically, let $v\in V(\Grid)$ be a grid-point, and consider the $1 \times 1$ square $[v-0.5, v+0.5)\times [v-0.5, v+0.5)$. If $\ell$ intersects this square, then we include $v$ in $V(\ell)$. Whenever $u,v\in V(\ell)$ are adjacent grid points, $\Highway(\ell)$ contains the edge $\{u,v\}\in E(\Grid)$.
If $a\in [-1,1]$ we let $w_\ell$ assign every edge in $\Highway(\ell)$ the weight $\frac{\sqrt{a^2+1}}{|a|+1}$,
which is the asymptotic ratio between the Euclidean distance of points on $\ell$ and the number of edges taken along the grid path $\Highway(\ell)$.  Otherwise we write $\ell$
as $x = a^{-1}(y-b)$ with $a^{-1}\in [-1,1]$ and use weight
$\frac{\sqrt{a^{-2}+1}}{|a^{-1}|+1}$.  All off-highway edges have weight $\infty$.

This weight assignment guarantees a discrepancy of at most 1 between Euclidean distances and grid distances along $\Highway(\ell)$.
The proof of \cref{lem:highway-approx} appears in the appendix.

\begin{lemma}
\label{lem:highway-approx}
    Let $\Highway(\ell)$ denote the highway that approximates a line $\ell$ of the form $y = ax + b$. Then for any two points $u, v \in V(\Highway(\ell))$, 
    $\left| \dist_{w_\ell}(u, v) - \|u - v\|_2 \right| \le 1$.
\end{lemma}

The highway transformation can also be applied to a line \emph{segment} $s$. 
We use the same notation $\Highway(s)$.

\subsection{The Hierarchical Highway Construction}
\label{subsec:hierarchical-construction}

We are trying to find a weight assignment for 
the finite grid $[n]\times [n]$ in order to prove \cref{thm:highway2}.

\paragraph{Parameters.} The construction is parameterized by $(k_i)_{1\leq i\leq m}$, where
\[
k_1 = \lfloor n^{1/5}\rfloor, \quad k_{i+1} = \lfloor k_i^{1/2} \rfloor,
\]
and $m$ is minimum such that $k_m < 100$.
    
\paragraph{Layers of Lines.}
    The construction is based on a hierarchical system of lines in $\mathbb{R}^2$ which will eventually be embedded as highways in the grid.  The lines at level $i$
    have angles selected from $(\theta_{i,j})_{0\leq j < k_i}$:
\[
\theta_{i,j} = \frac{\pi \cdot j}{k_i}.
\]
    Fixing $i$ and one such angle $\theta_{i,j}$, there
    are many lines with angle $\theta_{i,j}$, spaced at distance $k_i^4$.  
    For $i\in [1,m],j\in [0,k_i-1], t\in \mathbb{Z}$, 
\[
\ell_{i,j,t} = \left\{ (x, y) \in \mathbb{R}^2 \ \middle|\  y \cdot \cos \theta = x\cdot \sin \theta + t \cdot k_i^4 \right\}.
\]
    Define $\Lines[i] = \{\ell_{i,j,t}\}$ to be the set of all lines at level $i$.

We cannot choose a weight function $W(n)$ that agrees with $w_{\ell_{i,j,t}}$ for \emph{every} line $\ell_{i,j,t}$ due to intersections.  Our solution is to avoid this issue by removing \underline{all} line intersections, which introduces distortions in distances that must be bounded.

Below we define a procedure to remove parts of $\ell_{i,j,t}$, leaving a set of line \emph{segments} $\mathcal{L}_{i,j,t}$.  Define $\Segments[i] = \bigcup \mathcal{L}_{i,j,t}$ to be the set of all line segments at level $i$.  If $O$ is an object or collection of objects,
define 
\[
\Fat(O,\delta) = \{p \mid \exists q\in O \mbox{ such that }
                            \|p-q\|_2 \leq \delta\}
\]
to be all points within distance $\delta$ of $O$.
Specifically, $\Fat(\ell,\delta)$ is a strip if $\ell$ is a line, and a hippodrome if $\ell$ is a segment.

Once $\Segments[1],\ldots,\Segments[i-1]$ are constructed, we construct $\Segments[i]$ as follows. 
For each $\ell_{i,j,t}\in \Lines[i]$
        initialize 
        $\ell \gets \ell_{i,j,t}$ and proceed to remove parts of $\ell$ in Steps 1 and 2.
        
\begin{description}
    \item[Step 1.]
        Set
        $\ell \gets \ell - \Fat(\Lines[i]-\{\ell\}, k_i)$,
        i.e., we remove every part of $\ell$ within distance $k_i$ of any other line at level $i$.

    \item[Step 2.]
        For each segment $s\in \bigcup_{i'<i}\Segments[i']$ 
        such that $\Fat(s,k_i)\cap \ell \neq \emptyset$,
        let $\overline{AB} = \Fat(s,k_i)\cap \ell$. 
        If $\|A-B\|_2 \geq k_i$, set $\ell \gets \ell - \overline{AB}$.
        Otherwise, let $B'\in \ell_{i,j,t}$ be such that $\|A-B'\|=k_i$ and $B\in \overline{AB'}$,
        and set $\ell \gets \ell - \overline{AB'}$.
        See \cref{fig:highway-construction}.
        Define $\mathcal{L}_{i,j,t} = \ell$, and include 
        all line segments of $\mathcal{L}_{i,j,t}$ in $\Segments[i]$.
\end{description}

\begin{figure}
    \centering
    \begin{tabular}{c}
        \includegraphics[width=0.7\linewidth]{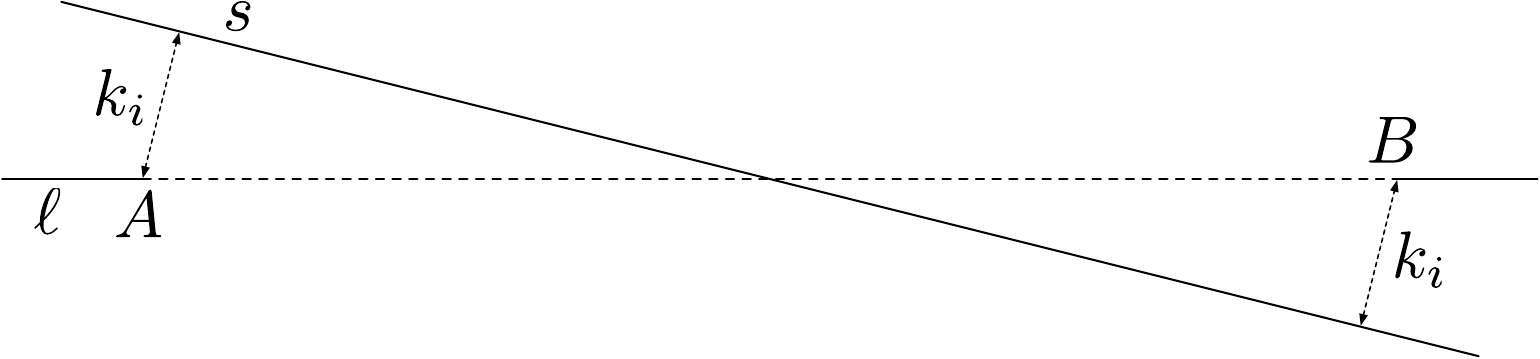}\\
        \textbf{(A)}\\
    \includegraphics[width=0.7\linewidth]{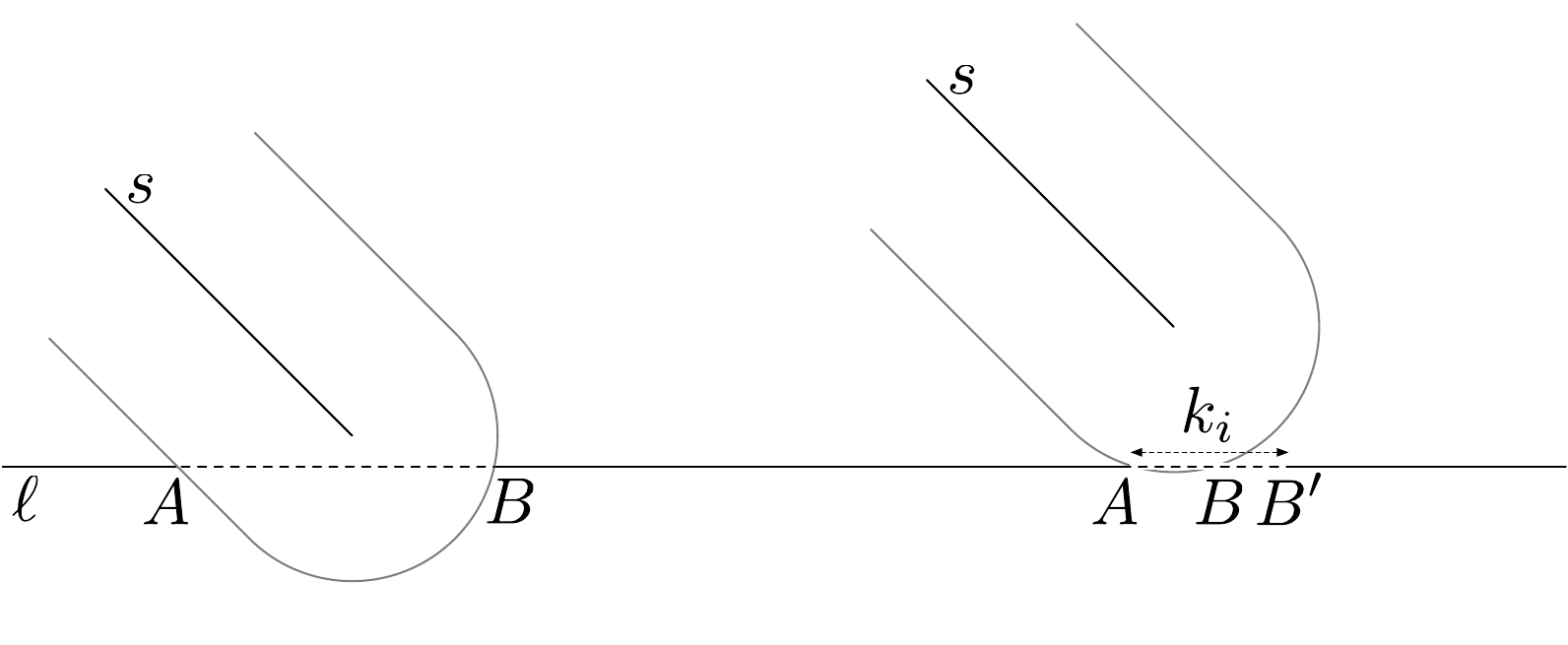}\\
         \textbf{(B)}
    \end{tabular}

    \caption{Illustrations of various cases in Step 2.  \textbf{(A)} $\ell$ and $s$ intersect. The segment $\ell\cap \Fat(s,k_i)$ is removed from $\ell$. \textbf{(B)} $\ell$ and $s$ do not intersect.  Left: $\overline{AB} \geq k_i$ and $\overline{AB}$ is removed from $\ell$.  Right: $\overline{AB} < k_i$ and $B'$ is such that $\overline{AB'}=k_i$ is removed from $\ell$.}
    \label{fig:highway-construction}
\end{figure}

The weight assignment $W(n)$ is now constructed as follows.  
For each line segment $s\in \bigcup_i \Segments[i]$, 
let $W(n)$ agree with $w_s$ at all edges in the corresponding highway segment $\Highway(s)$.  
All edges not appearing in any line segment have weight 2.

\begin{lemma}
\label{lem:highway-separation}
Fix any $s\in \Segments[i]$, $s'\in \Segments[i']$,
where $s\neq s'$ and $i\geq i'$.
For any points $u\in s,u'\in s'$, 
$\|u-v\|_2 \geq k_i$.  
\end{lemma}

\begin{proof}
%Steps 1 and 2 guarantee that $\mathcal{L}_{i,j,t}$ is a subset of $\ell_{i,j,t} - \Fat(\bigcup_{i'<i}\Segments[i']\cup \Lines[i], k_i)$, which implies that the distance from $u$ to any $u'$ in a lower level segment must be at distance at least $k_i$.

Steps 1 and 2 ensure that all remaining points on $\mathcal{L}_{i,j,t}$ lie outside 
\[
\Fat\left(\bigcup_{i'<i} \Segments[i'] \cup (\Lines[i] - \{\ell_{i,j,t}\}), k_i\right),
\]
and that any two consecutive segments on $\mathcal{L}_{i,j,t}$ are separated by distance at least $k_i$. This implies that the Euclidean distance between any $u \in s$ and $u' \in s'$ is at least $k_i$.
\end{proof}

\cref{lem:highway-lower-upper-bound} shows that distances under $W(n)$ are approximately Euclidean, up to a multiplicative stretch of $1+O(k_i^{-1})$ and additive stretch $O(k_i^4)$, for every index $i$.

\begin{lemma}
\label{lem:highway-lower-upper-bound}
For any $u, v\in [n]^2$ and $i\in [m]$.
\[
\|u - v\|_2 - 1 \le \dist_{W(n)}(u, v) \le \|u - v\|_2 + 
O\left(
k_i^{-1}\|u-v\|_2 + k_i^4
\right).
\]
\end{lemma}

\begin{proof}
Observe that the highways corresponding to all segments in $\bigcup_i \Segments[i]$ are vertex-disjoint.  Thus, every grid path $P$ from $u$ to $v$ 
can be written as $B_0A_1B_1A_2B_2\cdots A_kB_k$, where each $A_j \subset \Highway(s)$ for some $s\in \bigcup_i \Segments[i]$ and each $B_j$ is disjoint from all highway segments,
and therefore consists of only weight-2 edges. 
(One or both of $B_0,B_k$ may be empty.)
By \cref{lem:highway-approx}, $W(A_j)$ is at least the Euclidean distance between its endpoints minus one, and $W(B_j)$ is at least the Euclidean distance multiplied by 2,  which implies that $W(B_j)$ is at least the Euclidean distance plus 1 for $1\le j \le k-1$, so $W(P)$ is at least $\|u-v\|_2-1$.
\medskip 

Turning to the upper bound, We bound the distance $\dist_{W(n)}(u, v)$ by explicitly constructing a path that stays within the vicinity of a single line $\ell_{i,j,t}$.
Since level-$i$ lines occur at angular intervals of $\frac{\pi}{k_i}$ and parallel lines are spaced $k_i^4$ apart, 
we can always find an $\ell^* = \ell_{i,j,t}$ satisfying
the following properties.
First, the difference in angle between $\ell^*$ and $\overline{uv}$
is at most $\frac{\pi}{2k_i}$. 
Second, the distance from $u$ to $\ell^*$ is at most $k_i^4/2$.
Let $A$ and $B$ be the closest grid points on $\Highway(\ell^*)$
from $u$ and $v$ respectively.  It follows that
\begin{equation}\label{eqn:u-A-B-v}
\dist_{W(n)}(u,A) + 
\dist_{W(n)}(v,B) = 
O\left(k_i^4 + \|u,v\|_2\sin\left(\frac{\pi}{2k_i}\right)\right)
= O(k_i^4 + k_i^{-1}\|u-v\|_2).
\end{equation}

It remains to bound $\dist_{W(n)}(A,B)$.
A trivial upper bound is $\dist_{W(n)}(A,B) \leq 2 \cdot \sqrt{2}\|A-B\|_2$, so if $\|A-B\|_2 < 100 k_i^4$ we are done.
Henceforth we shall assume that
$\|A-B\|_2 \geq 100 k_i^4$.
We would prefer to follow the $A$-$B$ path along $\Highway(\ell^*)$,
but sections of this highway have effectively been removed by 
Steps 1 and 2 of the construction.  We bound the stretch induced 
by the gaps in the highway introduced in Steps 1 and 2 separately.

\paragraph{Step 1 Stretch.}  Fix a direction $\theta_{i,j'}$ different from $\ell^*$'s direction $\theta_{i,j}$.  Whenever $\ell^*$ intersects a line with angle $\theta_{i,j'}$, Step 1 causes
$\Highway(\ell^*)$ to lose $\frac{k_i}{\sin|\theta_{i,j}-\theta_{i,j'}|} = O(k_i^2)$ edges.  There are $k_i$ angles, and 
parallel lines with angle $\theta_{i,j'}$
are spaced $k_i^4$ apart, so the total number of edges removed
from the $A$-$B$ path in $\Highway(\ell^*)$ in Step 1 is
\begin{equation}\label{eqn:step1-stretch}
O\left(%
k_i^2\cdot k_i \cdot \frac{\|A-B\|_2}{k_i^4}
\right) = O(k_i^{-1} \|A-B\|_2).
\end{equation}
The additive stretch induced by walking across the gaps 
induced by Step 1 is also $O(k_i^{-1}\|A-B\|_2)$ as all 
these edges have weight 2.

\paragraph{Step 2 Stretch.}  Whenever part of $\Highway(\ell^*)$
is removed by Step 2 we do \emph{not} walk precisely in the direction of $\ell^*$ but take a \emph{detour} to a lower level highway.  Suppose that in Step 2, 
a segment $s\in \bigcup_{i'<i} \Segments[i']$ causes
an interval $\overline{CD}$ of $\ell^*$ to be removed.  
Define $E,F$ to be the points on $s$ closest to $C,D$, respectively.
When our path reaches $C$, we walk from $C$ to $E$, then to $F$ along $\Highway(s)$, then to $D$.  See \cref{fig:CDEF}.
By construction
$\|C-E\|_2, \|D-F\|_2 = O(k_i)$.
Since $\|E-F\|_2 \leq \|C-D\|_2$ and by \cref{lem:highway-approx}
$\dist_{W(n)}(E,F)=\|E-F\|_2\pm 1$,
the additive stretch due to the conflict with $s$ is at most
\[
\dist_{W(n)}(C,E)+\dist_{W(n)}(E,F)+\dist_{W(n)}(F,D) - \dist_{w_{\ell^*}}(C,D) 
= O(k_i).
\]

The last task is to bound the number of such segments interfering with the $A$-$B$ path.  Observe that $s$ is a segment 
of a line at level $i-1$ or lower.  Thus, by \cref{lem:highway-separation} any two such segments $s,s'$ are at distance at least $k_{i-1} \ge k_i^2$, and the
total additive stretch caused by Step 2 detours is
\begin{equation}\label{eqn:step2-stretch}
O\left(% 
k_i \cdot \frac{\|A-B\|_2}{k_i^2}
\right)
=
O(k_i^{-1}\|A-B\|_2).
\end{equation}

\begin{figure}
    \centering
    \begin{tabular}{c}
    \includegraphics[width=0.7\linewidth]{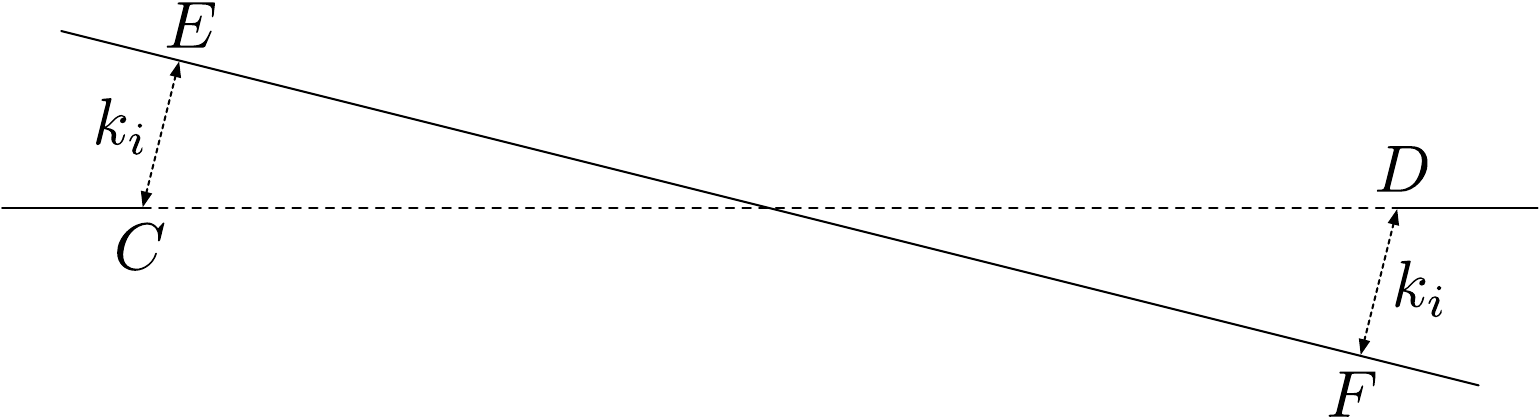}\\
\textbf{(A)}\\
    \includegraphics[width=0.7\linewidth]{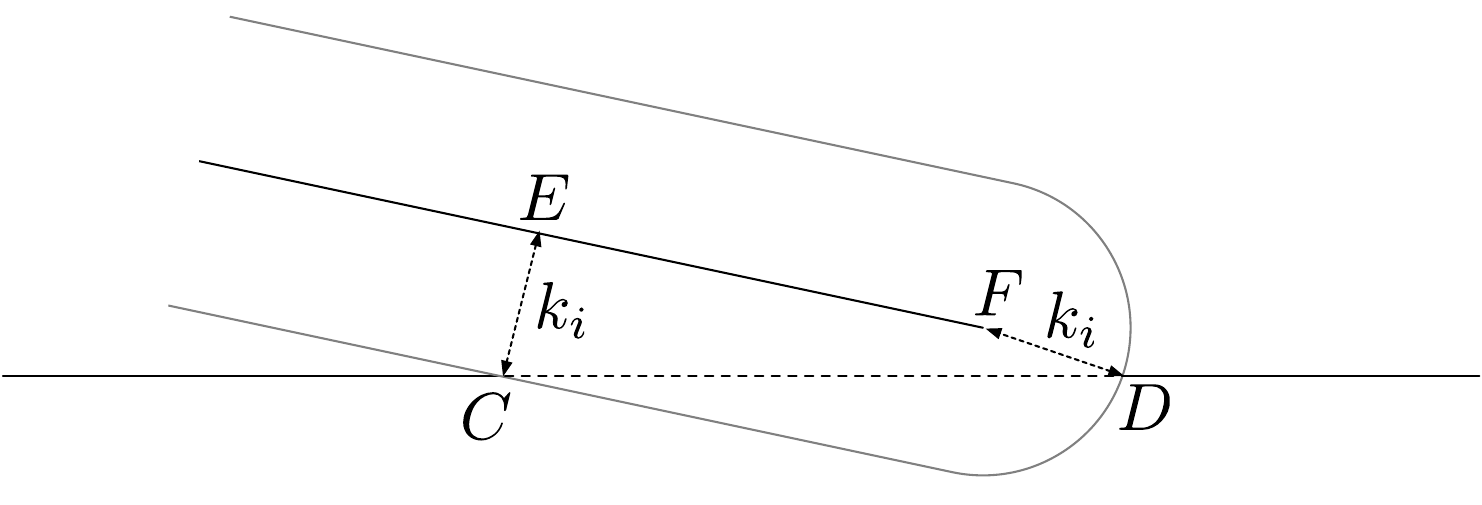}\\
\textbf{(B)}\\
    \end{tabular}
    \caption{Step 2 detours from the proof of \cref{lem:highway-lower-upper-bound}. \textbf{(A)} The case when segment $s$ intersects $\ell^*$.
    \textbf{(B)} When segment $s$ does not intersect  $\ell^*$, but $\Fat(s,k_i)$ does.}
    \label{fig:CDEF}
\end{figure}

Combining \cref{eqn:u-A-B-v,eqn:step1-stretch,eqn:step2-stretch}, we conclude that
\begin{align*}
    \dist_{W(n)}(u,v) 
    &\leq \dist_{W(n)}(u,A) + \dist_{W(n)}(A,B) + \dist_{W(n)}(B,v)\\
    &\leq \|u-v\| + O\left(k^{-1}\|A-B\|_2 + k_i^{-1}\|u-v\|_2 + k_i^4\right)\\
    &= \|u-v\|_2 + O\left(k^{-1}\|u-v\|_2 + k_k^4\right).
\end{align*}
\end{proof}

\begin{proofofthm}{\ref{thm:highway2}} Recall that by the definition of the sequence $(k_i)$, for any pair of points $u, v$ with Euclidean distance $d = \|u - v\|_2 > 100^9$, there exists some index $1 \le i \le m$ such that
\[
k_i \in \left[d^{1/9} - 1,\ d^{2/9}\right].
\]
Applying~\cref{lem:highway-lower-upper-bound}, we have
\begin{align*}
\dist_W(u, v) &\le \|u - v\|_2 + O\left( \frac{1}{k_i} \cdot \|u - v\|_2 + k_i^4 \right) \\
&\le \|u - v\|_2 + O\left( \|u - v\|_2^{8/9} \right).
\end{align*}
The theorem holds trivially when
$d \le 100^9$, which completes 
the proof of~\cref{thm:highway2}.
\end{proofofthm}

To prove \cref{thm:highway1} we give a ``black box'' reduction showing that
any construction that gives a bound
like \cref{thm:highway2} for the finite grid $[n]\times [n]$ yields the same guarantee on the infinite grid $\mathbb{Z}\times\mathbb{Z}$.

We begin by tiling the integer grid with various size squares as follows.  The central tile is $1000\times 1000$, which is surrounded by eight $1000\times 1000$ tiles, all of which, in turn, are surrounded by eight $3000\times 3000$ tiles, 
which are in turn surrounded by eight $9000\times 9000$ tiles, and so on.  See \cref{fig:highway2}.
Within each of these $n\times n$ tiles, we apply \cref{thm:highway2} to choose the weight function
in the central $(n-2)\times (n-2)$ grid.
All edges with at least one endpoint on the boundary of the tile have weight 2.  Let $W$ be the resulting weight function of $E(\Grid)$.

\begin{figure}
    \centerline{\includegraphics[scale=.15]{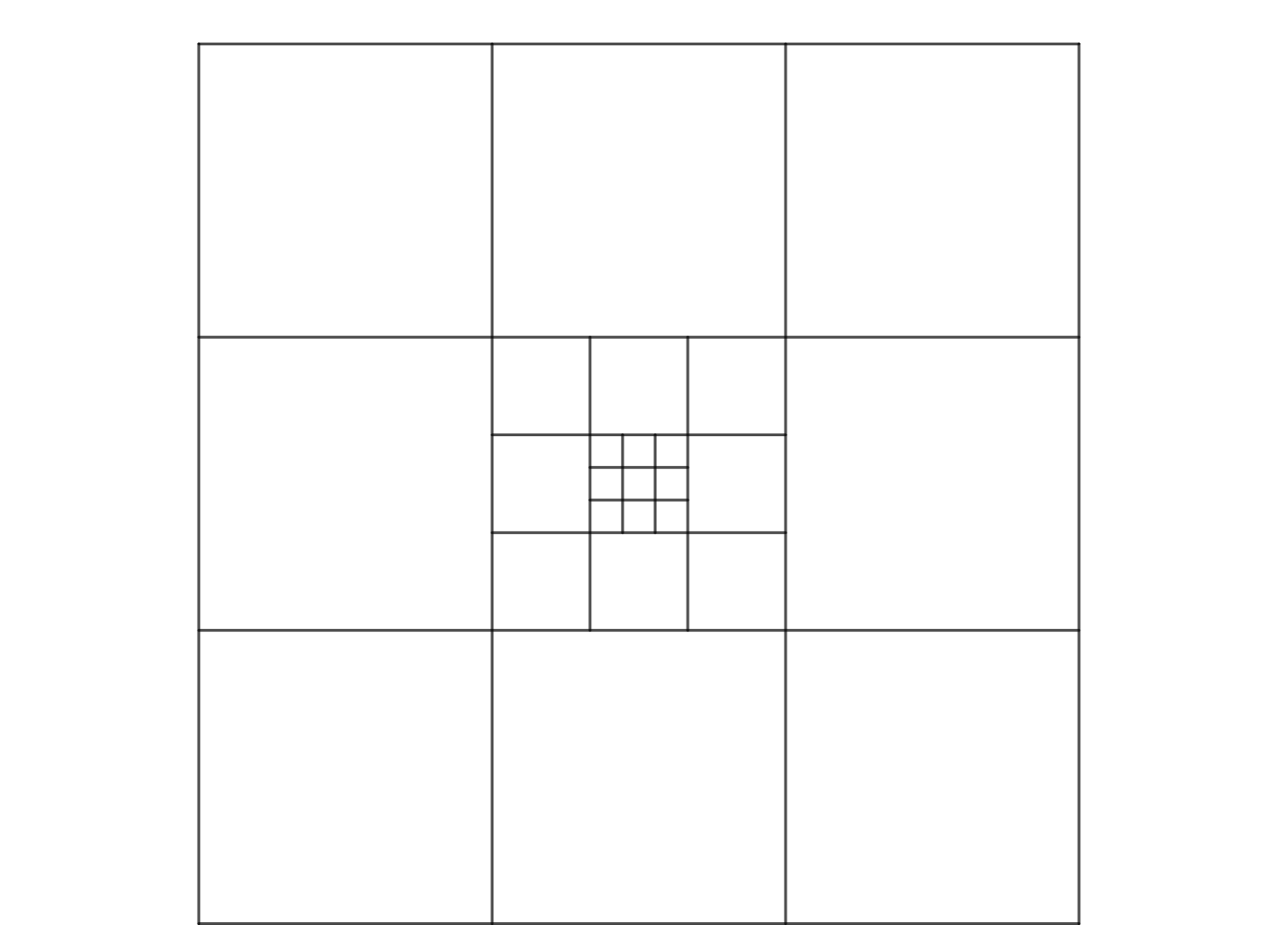}}
    \caption{\label{fig:highway2}Illustration of the recursive tiling: at each level, we place eight squares around the current one to expand the scale.}
\end{figure}

\begin{proofofthm}{\ref{thm:highway1}}
Consider any two $u,v\in V(\Grid)$ and let 
$L\geq 0$ be minimum such that 
$\|u-v\|_2 \leq 3^L 1000$.  The line $\overline{uv}$ can intersect at most 3 tiles with dimensions $3^L 1000$ or larger, and at most 4 tiles with dimension $3^i 1000$, $i<L$.  Thus, by concatenating shortest paths inside each tile, by \cref{thm:highway2},
the total additive stretch is at most
\begin{align*}
\dist_W(u,v) - \|u-v\|_2 
= 
O\left(\|u-v\|_2^{8/9} + \sum_{i=0}^{L-1}(\sqrt{2}\cdot 3^i 1000)^{8/9}\right) 
= 
O(\|u-v\|_2^{8/9}).
\end{align*}
The same argument from \cref{lem:highway-approx} shows that $\dist_W(u,v) \geq \|u-v\|_2 -1$.
\end{proofofthm}

\section{Experimental Findings in the Squishy Grid}\label{sect:experimental-findings}

If \cref{q:randomized-squishy-grid} seems too daunting, a natural idea is to simplify the problem by considering only \emph{monotone} paths, that is, paths that use the fewest number of edges.  

\begin{question}\label{q:monotone-squishy-grid}
Define $\monodist_w(u,v)$ to be the length of the shortest $u$-$v$ 
path that uses $\|u-v\|_1$ edges.   Does there exist a distribution $\mathscr{D}^*$ over $\R_{\geq 0}$
such that if $\Grid[w] \sim \Grid[\mathscr{D}^*]$ 
is a randomly weighted graph, for all $u,v\in V(\Grid)$,
\[
\E(\monodist_w(u,v)) = (1\pm o(1))\|u-v\|_2.
\]
\end{question}

At first glance this problem may seem easier, or more plausible, than \cref{q:randomized-squishy-grid}.  
Whereas it is an open problem finding a distribution $\mathscr{D}$ satisfying \cref{eqn:0-45} (the time constant in the $0^{\circ}$ and $45^{\circ}$ directions are 1), this is nearly trivial when we consider $\monodist$.
\begin{lemma}\label{lem:0-45-monodist}
There exists a distribution $\mathscr{D}$ on $\R_{\geq 0}$ such that
\[
\lim_{n\to \infty} 
\frac{\E(\monodist_w(\mathbf{0}, ne_0))}{n} = \lim_{n\to \infty} \frac{\E(\monodist_w(\mathbf{0}, ne_{45}))}{n}=1.
\]
\end{lemma}

\begin{proof}
    Since there is only one path from $\mathbf{0}$ to $ne_0$, 
    any distribution $\mathscr{D}$ with $\E(w_0\sim \mathscr{D}) = 1$ works for the $0^\circ$ direction.  Consider the class of distributions $\mathscr{D}[\epsilon]$, where 
    $\Pr(w_0 = 1-\epsilon)=\Pr(w_0 = 1+\epsilon)=1/2$.  
    When $\Grid[w]\sim\Grid[\mathscr{D}[0]]$, 
    $\monodist_w(\mathbf{0},ne_{45}) = \sqrt{2}\cdot \|u-v\|_2$.
    We argue that when $\Grid[w]\sim\Grid[\mathscr{D}[1]]$,
    $\E[\monodist_w(\mathbf{0},ne_{45})] < \|u-v\|_2/\sqrt{2} + O(\sqrt{n})$.
    When $\epsilon=1$ all weights are 0 or 2 with equal probability.
    We walk myopically from the origin, taking a weight-0 edge North or East whenever possible, or a weight-2 edge North or East if necessary, until we reach a barrier when the $x$- or $y$-coordinate matches $ne_{45}$.  When the edges in both directions have the same weight, we choose one randomly. Before reaching a barrier, the expected weight of the next edge is $(3/4)\cdot 0 + (1/4)\cdot 2 = 1/2$ and after reaching a barrier it is $(1/2)(0+2)=1$.  There are $O(\sqrt{n})$ edges in the latter category, in expectation,
    so $\E(\monodist_w(\mathbf{0},ne_{45})) < (1/2)\|ne_{45}\|_1 + O(\sqrt{n})=n/\sqrt{2} + O(\sqrt{n})$.  
    By the intermediate value theorem, there has to be some $\epsilon^* \in [0,1]$
    such that
    \[
    \lim_{n\to \infty} 
\frac{\E(\monodist_w(\mathbf{0}, ne_0))}{n} = \lim_{n\to \infty} \frac{\E(\monodist_w(\mathbf{0}, ne_{45}))}{n}=1.
    \]
\end{proof}

Let $B_{\mono}(t) = \{u\in \Z^2 \mid \monodist_w(\mathbf{0},u)\leq t\}$
and $\mathcal{B}_{\mono}(\mathscr{D})$ be the limiting shape 
of $B_{\mono}(t)/t$ in $\Grid[w]\sim\Grid[\mathscr{D}]$ as $t\to \infty$.
Thus, $\mathcal{B}_{\mono}(\mathscr{D}_{\epsilon^*})$ coincides with the 
$L_2$-ball in the eight intercardinal directions N, E, S, W, NE, SE, SW, NW.
If $\mathcal{B}_{\mono}(\mathscr{D}_{\epsilon^*})$ were convex, 
then it would have to be quite close to the $L_2$-ball. 

Unfortunately, our experiments show that $\mathcal{B}_{\mono}(\mathscr{D}_{\epsilon^*})$ is \emph{not} convex,
which casts serious doubt on \cref{q:monotone-squishy-grid} having an affirmative answer.
See \cref{fig:monotone-stretch-results}.

\begin{figure}
    \centering
    \begin{tabular}{cc}
    \includegraphics[width=0.45\linewidth]{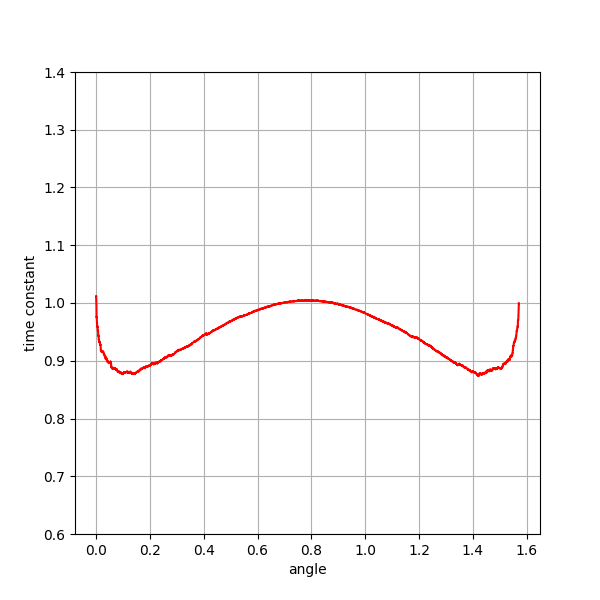}
&
    \includegraphics[width=0.45\linewidth]{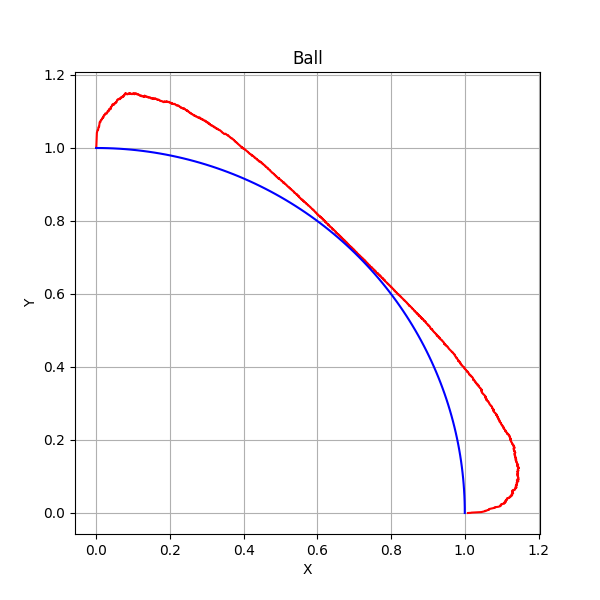}\\
    \textbf{(A)} & \textbf{(B)}
    \end{tabular}
    \caption{\textbf{(A)} 
    The stretch of $\monodist_w(\mathbf{0},ne_\theta)/\|ne_\theta\|_2$, as a function of the angle 
    $\theta \in [0,\frac{\pi}{2}]$, 
    expressed in radians.
    \textbf{(B)} The shape of $\mathcal{B}_{\mono}(\mathscr{D}_{\epsilon^*})$.}
    \label{fig:monotone-stretch-results}
\end{figure}

In retrospect, \cref{q:monotone-squishy-grid} is less likely than \cref{q:randomized-squishy-grid}
to be answered in the affirmative
since $\monodist_w(\mathbf{0},(n,m))$, $m\leq n$,  
is much more sensitive to small deviations in $m$ than $\dist_w(\mathbf{0},(n,m))$
Considering the cases when $m=0, m=0.1n$, and $m=n$, $\monodist_w(\mathbf{0},(n,m))$ is the minimum of
${n + 0\choose 0}=1, {n+0.1n\choose 0.1n}\approx 1.39^n$, and ${2n\choose n}\approx 4^n$
different paths, respectively.
This sharp jump from constant to exponential in the vicinity $m=0$ does not exist in \cref{q:randomized-squishy-grid}.
Assuming the variance of $\mathscr{D}$ is sufficiently large,
$\dist_w(\mathbf{0},ne_0)$ is the minimum of an exponential number 
of plausible shortest paths.

\subsection{Discrete Distributions for First Passage Percolation}

When dealing with discrete distributions the most natural measure of complexity is \emph{support size}.  Therefore, we study \cref{q:randomized-squishy-grid} experimentally by considering the space of 2- and 3-point distributions.
For a fixed integer $k$,
the $k$-point distribution
$\mathscr{D}(\lbrace (p_i,x_i)\rbrace_{i=1}^{k})$, 
is such that
\[
\Pr_{w_0\sim \mathscr{D}}(w_0 = x_i) = p_i.
\]
It is determined by $2k-1$ parameters, 
as $p_k=1-(p_1+\cdots+p_{k-1})$.

\subsubsection{Experimental Methodology}

To identify locally optimal distributions in the space of
$k$-point discrete distribution $\mathscr{D}(\{(p_i, x_i)\}_{i=1}^k)$, 
we employ a two-layer iterative strategy:

\begin{itemize}
    \item We first perturb the probability vector $(p_1, \ldots, p_k)$.
    \item For fixed probabilities $(p_1, \ldots, p_k)$, we generate $k$ random initial values $x_1, x_2, \ldots, x_k$ and then alternate between the following two update steps:
    \begin{description}
        \item[Perturbation Step.] We perturb each value $x_i$ and compute the estimated directional stretch for both $\theta = 0$ and $\theta = \frac{\pi}{4}$. 
        Let $\mu_\theta$ be the empirical 
        ratio $\frac{\dist_w(\mathbf{0}, ne_\theta)}{n}$, obtained from this round of simulation.  Here $n\approx $ 30,000.
        \item[Normalization Step.] We normalize the values $\{x_i\}_{i=1}^k$ by setting 
        \[
        x_i \leftarrow \frac{x_i}{\sqrt{\mu_0 \cdot \mu_{\frac{\pi}{4}}}}.
        \]
        This scaling ensures that the average stretch along the cardinal and intercardinal directions remains close to 1, thereby facilitating comparisons between distributions.
    \end{description}
\end{itemize}

\subsubsection{2-Point Distributions}

The best $2$-point distribution identified with this method is 
$\mathscr{D}_2$, given below.  Roughly speaking, every edge weight is either $0.41$ or $4.75$, $44$\% and $56$\% 
of the time, respectively.
\[
\mathscr{D}_2 = \{(0.44273, 0.41401), (0.55727, 4.75309)\}.
\]
We find that $\Grid[w]\sim \Grid[\mathscr{D}_2]$ empirically approximates Euclidean distances 
up to stretch $1.00750$, i.e., up to $3/4\%$ error.
\cref{fig:2-point-distribution-results}(A) plots the observed stretch $\dist_w(\mathbf{0},ne_\theta) / \|ne_\theta\|_2$ as a function of the angle $\theta \in [0,\pi/2)$. 
Figure ~\cref{fig:2-point-distribution-results}(B) shows the set of all grid points whose empirical graph distance from the origin first exceeds $n$. 
The resulting boundary is visually 
close to a Euclidean circle, 
suggesting that $\mathscr{D}_2$ induces 
an approximately isotropic metric in expectation.

\begin{figure}
    \centering
    \begin{tabular}{cc}
    \includegraphics[width=0.45\linewidth]{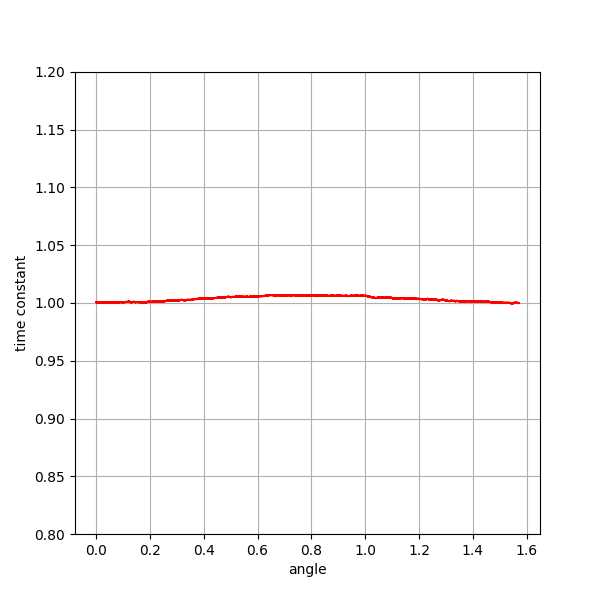}
&
    \includegraphics[width=0.45\linewidth]{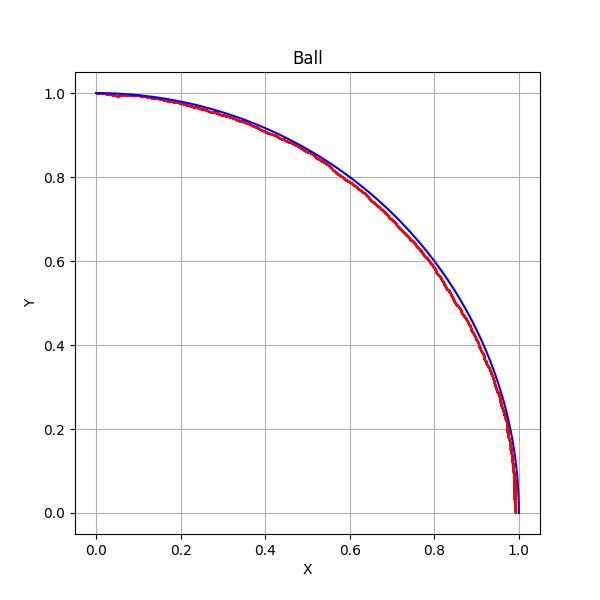}\\
    \textbf{(A)} & \textbf{(B)}
    \end{tabular}
    \caption{Results on the 2-point distribution $\mathscr{D}_2$.\textbf{(A)} Stretch $\dist_w(\mathbf{0},ne_\theta) / \|ne_\theta\|_2$, as a function of the angle $\theta \in [0,\frac{\pi}{2}]$.
    \textbf{(B)} The empirical distance-$n$ ball in $\Grid[w]\sim\Grid[\mathscr{D}_2]$.}
    \label{fig:2-point-distribution-results}
\end{figure}

Note that $\E(w_0\sim \mathscr{D}_2) \approx 2.83$, meaning that $\dist_w(\mathbf{0},ne_0)\approx n$
is likely to be realized by a highly non-monotone path, consistent with the observations in \cref{fig:monotone-stretch-results}.  Two sample paths from the origin to $(1000, 0)$ and $(1000, 100)$ are shown in \cref{fig:2-point-distribution-path}.

\begin{figure}
    \centering
    \begin{tabular}{cc}
    \includegraphics[width=0.45\linewidth]{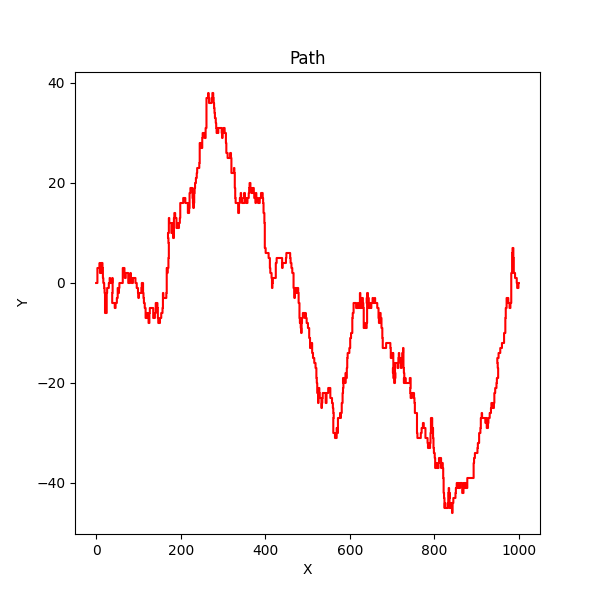}
&
    \includegraphics[width=0.45\linewidth]{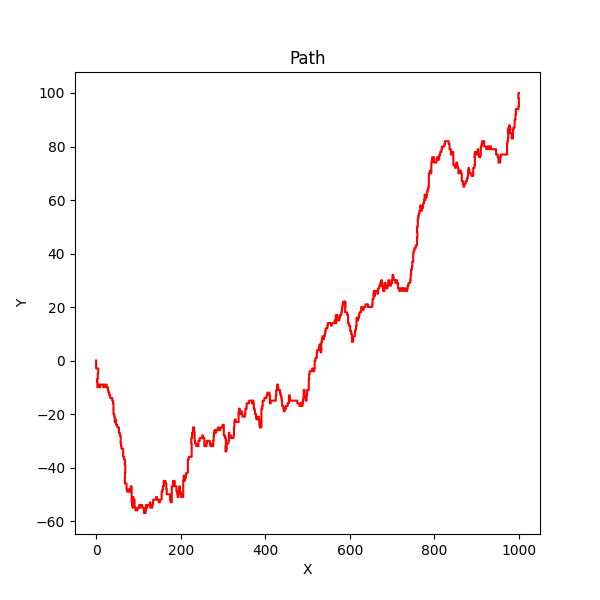}\\
    \textbf{(A)} & \textbf{(B)}
    \end{tabular}
    \caption{Results on the 2-point distribution $\mathscr{D}_2$.
    \textbf{(A)} The trace of a shortest path from $(0,0)$ to $(1000,0)$.
    \textbf{(B)} The trace of a shortest path from $(0,0)$ to $(1000,100)$.}
    \label{fig:2-point-distribution-path}
\end{figure}

\subsubsection{3-Point Distributions}

$\mathscr{D}_2$ does not leave much room for improvement, 
but we are able to eke out a slightly better empirical stretch of 1.00622 with a 3-point distribution $\mathscr{D}_3$.
\begin{align*}
    \mathscr{D}_3 &=\mathscr{D}(\{(0.34809,0.20647),(0.25735,2.51586),(0.39456,9.32215)\}).
\end{align*}
See \cref{fig:3-point-distribution-results} 
for visual representations of the empirical stretch of $\mathscr{D}_3$. 

\begin{figure}
    \centering
    \begin{tabular}{cc}
    \includegraphics[width=0.4\linewidth]{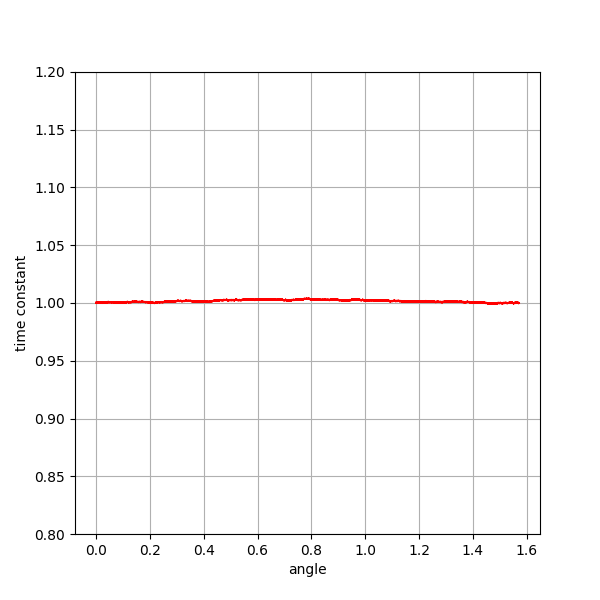}
&
    \includegraphics[width=0.4\linewidth]{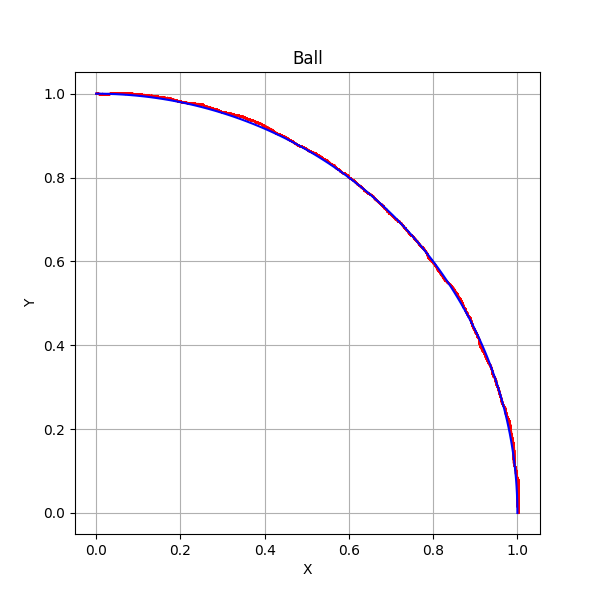}\\
    \textbf{(A)} & \textbf{(B)}\\

    \end{tabular}
    \caption{Results on the 3-point distribution $\mathscr{D}_3$. 
    \textbf{(A)} Stretch $\frac{\dist_w(\mathbf{0},ne_\theta)}{\|ne_\theta\|_2}$, as a function of $\theta \in [0,\frac{\pi}{2}]$.
    \textbf{(B)} The empirical distance-$n$ ball in $\Grid[w]\sim\Grid[\mathscr{D}_3]$.}
    \label{fig:3-point-distribution-results}
\end{figure}

\subsection{$L_p$-Balls and Continuous Distributions}

Alm and Deijfen~\cite{AlmD15} experimented with many of the standard
continuous distributions, such as uniform, exponential, Gamma, and Fisher.  Only a distribution from the Fisher class approximated Euclidean distances to within 1\%.  In this section we 
replicate some of Alm and Deijfen's findings, but instead of measuring error with respect to the $L_2$-norm, we show they are very good approximations for other $L_p$-norms, $p<2$.
\cref{fig:continuous-distribution-Lp} shows that (suitable scaled versions of) $\mathrm{Uniform}(0,1)$, $\Gamma(2,2)$ and $\Gamma(10,10)$ are good approximations to the 
$L_{1.87},L_{1.85},$ and $L_{1.32}$ metrics, respectively. 
It is not true that every 
$\mathcal{B}(\mathscr{D})$ approximates an $L_p$-ball.
For some non-constant distributions, 
the limit shape $\mathcal{B}(\mathscr{D})$ has flat edges; see~\cite[\S 2.5]{AuffingerDH17}.

\begin{figure}
\centering
\begin{tabular}{ccc}
\includegraphics[width=0.3\linewidth]{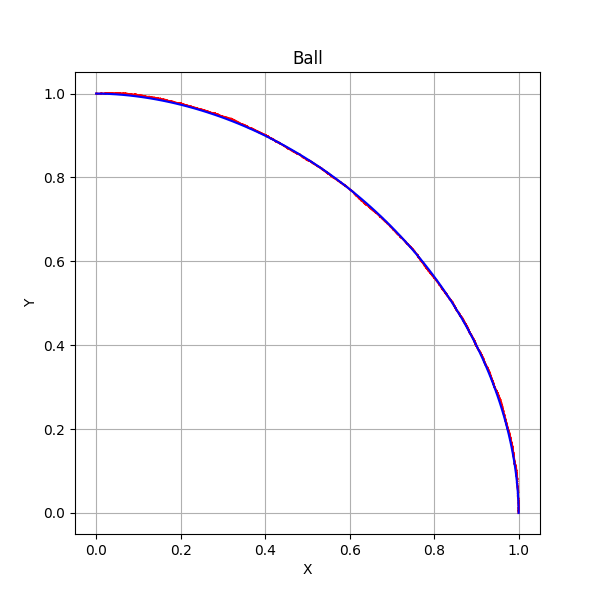} &
\includegraphics[width=0.3\linewidth]{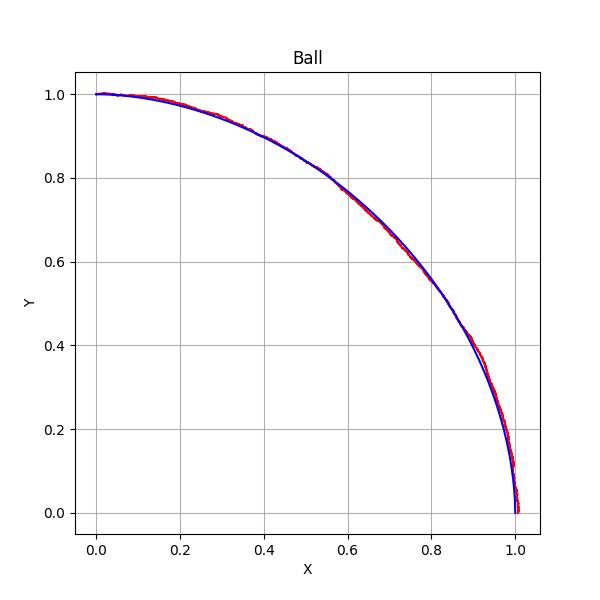} &
\includegraphics[width=0.3\linewidth]{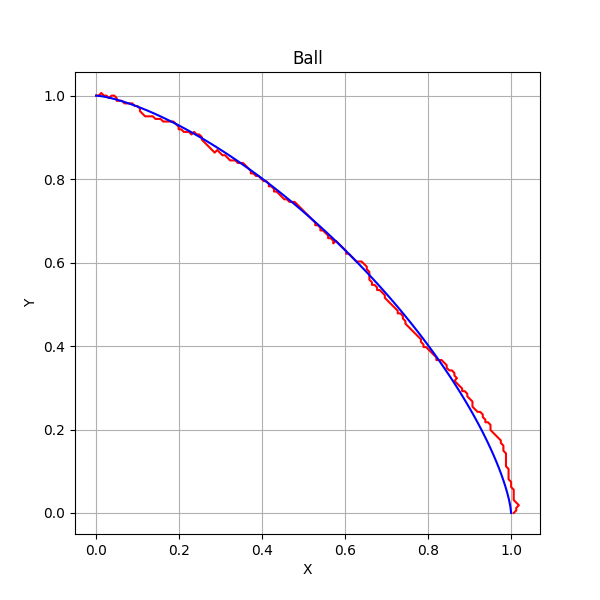} \\
\textbf{(A)} & \textbf{(B)} & \textbf{(C)}\\
\end{tabular}
\caption{Other $L_p$ balls, $p\in (1,2)$, and rescaled versions of some continuous distributions. (A) $L_{1.87361}$ ball and $\mathrm{Uniform}(0,1)$ (B) $L_{1.85691}$ ball and $\Gamma(2,2)$ (C) $L_{1.32879}$ ball and $\Gamma(10,10)$\label{fig:continuous-distribution-Lp}}
\end{figure}

We measure the empirical error of the graph distances in $\Grid[w]\sim \Grid[\mathscr{D}]$ 
with respect to an $L_p$-norm as follows.
\[
\Err(p;w) = \max_{\boldu\in \Ball_{n}} \left\lvert \frac{ \lVert \boldu\rVert_2}{\lVert \boldu_0\rVert_2}-\left(|\cos(\theta(\boldu))|^p+|\sin(\theta(\boldu))|^p\right)^{-\frac{1}{p}}\right\rvert.
\]
Here $n$ is a large constant,
$\Ball_n$ is the set of all $\boldu$
such that $\dist_w(\mathbf{0},\boldu)\geq n$ but the immediate predecessor of $\boldu$ is at distance less than $n$,
and $\boldu_0$ is the first vertex of $\Ball_n$ on the $x$-axis.  The 
angle of $\boldu$ is $\theta(\boldu)$.
Intuitively, if the time constant of $\mathscr{D}$ is not 1, we can effectively make it 1 for these calculations by
normalizing by $\|\boldu_0\|$.
Then $\Err(p;w)$ measures the maximum 
difference between the
normalized distance to a $\boldu\in\Ball_n$ and 
its idealized radial profile in the 
$L_p$ ball, evaluated in direction $\theta(\boldu)$.

Our experimental results show that for some continuous distributions $\mathscr{D}$, $\Grid[w]\sim\Grid[\mathscr{D}]$
closely approximates an $L_p$-metric 
for $p<2$.

\begin{center}
\begin{tabular}{|l|l|l|}
\hline
\textbf{Distribution $\mathscr{D}$} & \textbf{$p$ value} & \textbf{$\Err(p;w)$} \\
\hline
$\mathrm{Uniform}(0,1)$ & 1.87361 & 0.00462 \\
\hline
$\Gamma(2,2)$ & 1.85691 & 0.00986 \\
\hline
$\Gamma(10,10)$ & 1.32879 & 0.03658 \\
\hline
\end{tabular}
\end{center}

\section{Conclusion}\label{sect:conclusion}

In this paper we asked how well the integer grid graph $\Grid$ can approximate Euclidean distances, if weighted appropriately.  We gave two deterministic weighting schemes that answer \cref{q:squishy-grid}, the best one achieving a \textbf{\emph{polynomial}} additive stretch, that is, $\E[\dist_w(u,v)] = \|u-v\|_2 + O(\|u-v\|_2^{1-\delta})$ for $\delta=1/9$. 
Improving the additive stretch to something
\textbf{\emph{subpolynomial}} $\|u-v\|_2^{o(1)}$
seems to require a new approach to the problem.  Our ``highway'' method seems incapable of achieving subpolynomial additive error, and even if \cref{q:randomized-squishy-grid} is answered affirmatively (choosing weights i.i.d.~from some distribution $\mathscr{D}^*$), the tail bounds here only give 
polynomial additive error~\cite{AuffingerDH17}.

\medskip 

We conjecture that no weighting achieves subpolynomial additive error.

\begin{conjecture}\label{conj:no-constant-error}
The additive error $\dist_w(u,v)-\|u-v\|_2$ of $w : E(\Grid)\to \mathbb{R}_{\geq 0}$ 
is a function of $d = \|u-v\|_2$.  
\begin{description}
    \item[Weak Conjecture.] There is no $w$ with constant additive error $O(1)$, independent of $d$.
    \item[Strong Conjecture.] There is no $w$ with  subpolynomial additive error $d^{o(1)}$.
\end{description}
\end{conjecture}

\medskip 

In the randomized setting, \cref{q:randomized-squishy-grid} is equivalent to an old problem in \emph{first passage percolation}~\cite{AuffingerDH17,HammersleyW65} that is 
unsolved, but has some empirical evidence in its favor.  
Alm and Deijfen~\cite{AlmD15} showed that when $\mathscr{D}$ is a certain Fisher distribution, that $\Grid[w]\sim\Grid[\mathscr{D}]$ approximates Euclidean distances to less than $1\%$ error. 
In this paper, we demonstrated that a simple 2-point distribution $\mathscr{D}_2$ achieves $0.75\%$ error, and that a 3-point distribution $\mathscr{D}_3$ achieves 
$0.622\%$ error.  
As a practical matter, choosing weights according to 
$\mathscr{D}_2$ or $\mathscr{D}_3$ will induce smaller distance errors than our deterministic schemes, 
for all but extraordinarily large distances.

It is an interesting open problem to \underline{\emph{prove}} that in $\Grid[w]\sim \Grid[\mathscr{D}_2]$, the expected error of 
$\dist_w$ is at most 1\%, that is:
\[
\E(\dist_w(u,v)) = (1.005 \pm 0.005 \pm o(1))\|u-v\|_2.
\]

\paragraph{Acknowledgments.} We would like to thank Greg Bodwin for posing
the problem to us, G\'{a}bor Tardos and Boris Bukh 
for pointing us to references~\cite{PachPS90,RadinS96}
and~\cite{BuragoI15}, respectively,
and Wesley Pegden for 
suggesting we look into the percolation literature.

%\bibliographystyle{alpha}
%\bibliography{bibliography}

\begin{thebibliography}{ADH17}

\bibitem[AD15]{AlmD15}
Sven~Erick Alm and Maria Deijfen.
\newblock First passage percolation on $\mathbb{Z}^2$: A simulation study.
\newblock {\em J. Stat. Phys.}, 161:657--678, 2015.

\bibitem[ADH17]{AuffingerDH17}
Antonio Auffinger, Michael Damron, and Jack Hanson.
\newblock {\em 50 Years of First-Passage Percolation}.
\newblock University Lecture Series. American Mathematical Society, 2017.

\bibitem[Bak76]{baker1977-exposition}
Alan~M. Baker.
\newblock {T}he theory of linear forms in logarithms.
\newblock In {\em Transcendence theory: advances and applications (A. Baker, D. W. Masser, Eds.)}, pages 1--27. Academic Press, London, 1976.

\bibitem[BE15]{BorradaileE15}
Glencora Borradaile and David Eppstein.
\newblock Near-linear-time deterministic plane {S}teiner spanners for well-spaced point sets.
\newblock {\em Comput. Geom.}, 49:8--16, 2015.

\bibitem[BI15]{BuragoI15}
Dmitri Burago and Sergei Ivanov.
\newblock Uniform approximation of metrics by graphs.
\newblock {\em Proc. Amer. Math. Soc.}, 143:1241--1256, 2015.

\bibitem[Bug12]{weyl-thm}
Yann Bugeaud.
\newblock {\em Distribution Modulo One and Diophantine Approximation}, pages 1--14.
\newblock Cambridge University Press, 2012.

\bibitem[CD81]{CoxD81}
J.~Theodore Cox and Richard Durrett.
\newblock Some limit theorems for percolation with neces- sary and sufficient conditions.
\newblock {\em Annals of Probability}, 9:583--603, 1981.

\bibitem[CKT22]{ChangKT22}
Hsien{-}Chih Chang, Robert Krauthgamer, and Zihan Tan.
\newblock Almost-linear $\epsilon$-emulators for planar graphs.
\newblock In {\em Proceedings of the 54th Annual {ACM} {SIGACT} Symposium on Theory of Computing (STOC)}, pages 1311--1324, 2022.







\bibitem[CR98]{ConwayR98}
John~H. Conway and Charles Radin.
\newblock Quaquaversal tilings and rotations.
\newblock {\em Inventiones Mathematicae}, 132:179--188, 1998.

\bibitem[HW65]{HammersleyW65}
John~M. Hammersley and Dominic J.~A. Welsh.
\newblock First-passage percolation, subadditive processes, stochastic networks, and generalized renewal theory.
\newblock In {\em Bernoulli 1713, Bayes 1763, Laplace 1813: Anniversary Volume, Proceedings of an International Research Seminar, Statistical Laboratory, University of California, Berkeley, 1963}, pages 61--110. Springer, 1965.


\bibitem[Kes86]{Kesten86}
Harry Kesten.
\newblock Aspects of first passage percolation.
\newblock In {\em \'{E}cole d'\'{E}t\'{e} de Probabilit\'{e}s de Saint Flour XIV},
\newblock Lecture Notes in Mathematics, 1180.  Springer, 1986.

\ignore{
\bibitem[Kes86]{Kesten86}
Harry Kesten.
\newblock Aspects of first passage percolation.
\newblock In {\em \'{E}cole dÕ\'{E}t\'{e} de Probabilit\'{e}s de Saint Flour XIV, Lecture Notes in Mathematics, 1180}. Springer, 1986.
}

\bibitem[KN74]{epsilon-cover}
Lauwerens Kuipers and Harald Niederreiter.
\newblock {\em Uniform Distribution of Sequences}, pages 118--132.
\newblock Wiley, 1974.

\bibitem[PPS90]{PachPS90}
J{\'a}nos Pach, Richard Pollack, and Joel Spencer.
\newblock Graph distance and {E}uclidean distance on the grid.
\newblock In {\em Topics in Combinatorics and Graph Theory: Essays in Honour of Gerhard Ringel}, pages 555--559. Springer, 1990.

\bibitem[Rad94]{Radin94}
Charles Radin.
\newblock The pinwheel tilings of the plane.
\newblock {\em Annals of Mathematics}, 139:661--702, 1994.

\bibitem[RS96]{RadinS96}
Charles Radin and Lorenzo Sadun.
\newblock The isoperimetric problem for pinwheel tilings.
\newblock {\em Commun. in Math. Phys.}, 177:255--263, 1996.

\end{thebibliography}

\appendix

\section{Proofs}

\subsection{Proof of \cref{lem:highway-approx}}

\begin{proof}[Proof of \cref{lem:highway-approx}]
Without loss of generality, assume that $u = \mathbf{0}$ and $v = (x, y)$ for $x,y > 0$, and that the slope of the line $\ell$ is $a \in [0, 1]$, as other cases follow by symmetry. It follows from the definition of $\Highway(\ell)$ that the vertical deviations $|\ell(0)-0|$ and $|\ell(x)-x|$ are each bounded by $(1+|a|)/2$, and therefore we have $|ax-y|\leq 1+a$.

\begin{align*}
    \left| \dist_{w_\ell}(u, v) - \|u - v\|_2 \right| &= \left| (x+y)\frac{\sqrt{a^2+1}}{a+1} - \sqrt{x^2 + y^2}\right|\\
    &= \left| \frac{(x+y)\sqrt{a^2+1} - \sqrt{x^2 + y^2}(a+1)}{a+1}\right|\\
    &= \left| \frac{\left((x+y)\sqrt{a^2+1} \right)^2 - \left(\sqrt{x^2+y^2}(a+1)\right)^2}{(a+1)[(x+y)\sqrt{a^2+1}+(a+1)\sqrt{x^2+y^2}]} \right|\\
    &= \left| \frac{(x+y)^2(a^2+1)-(a+1)^2(x^2+y^2)}{(a+1)\left[(x+y)\sqrt{a^2+1}+(a+1)\sqrt{x^2+y^2}\right]} \right|\\
    &= \left| \frac{2(ax-y)(ay-x)}{(a+1)\left[(x+y)\sqrt{a^2+1}+(a+1)\sqrt{x^2+y^2}\right]} \right|\\
    &\le \frac{2(1+a)x}{(1+a)[x+x]}\\
    &= 1.
\end{align*}

\end{proof}

\end{document}